\documentclass[11pt]{article}
\usepackage{graphicx} 
\usepackage{fullpage}
\usepackage{amsfonts}
\usepackage{xcolor}
\usepackage{hyperref}
\usepackage{amsmath,amssymb,amsthm}
\usepackage{cleveref}
\usepackage{enumerate}
\usepackage{thm-restate}
\usepackage[T1]{fontenc}
\usepackage[colorinlistoftodos,prependcaption,textsize=small]{todonotes}

\newcommand{\eps}{\varepsilon}
\renewcommand{\angle}{\sphericalangle}
\newcommand{\slope}{\mathrm{slope}}

\definecolor{BrickRed}{rgb}{0.8, 0.25, 0.33}
\def\EMPH#1{\emph{\textcolor{BrickRed}{#1}}}

\theoremstyle{plain}
\newtheorem{theorem}{Theorem}
\newtheorem{lemma}[theorem]{Lemma}

\newtheorem{definition}[theorem]{Definition}

\newtheorem{problem}[theorem]{Problem}
\usepackage{palatino}

\hypersetup{
  colorlinks=true,
  linkcolor=red!70!black,
  linktoc=all,
  citecolor=purple
}

\title{Euclidean Noncrossing Steiner Spanners of\\ Nearly
Optimal Sparsity}
\author{
Sujoy Bhore \thanks{Department of Computer Science \& Engineering, Indian Institute of Technology Bombay, Mumbai, India. Work supported in part by ANRF ARG-MATRICS, Grant 002465. Email: \texttt{sujoy@cse.iitb.ac.in}}\\
\and Sándor Kisfaludi‑Bak \thanks{Aalto University, Espoo, Finland. Supported by the Research Council of Finland, Grant 363444. Email: \texttt{sandor.kisfaludi-bak@aalto.fi}}\\
\and Lazar Milenković \thanks{Tel Aviv University, Israel. Funded by a grant from the United States-Israel Binational Science Foundation (BSF), Jerusalem, Israel, and the United States National Science Foundation (NSF). Email: \texttt{milenkovic.lazar@gmail.com}}\\
\and Csaba D. Tóth \thanks{Department of Mathematics, California State University Northridge, Los Angeles, CA; and Department of Computer Science, Tufts University, Medford, MA, USA.  Research supported, in part, by the NSF award DMS-2154347. Email: \texttt{csaba.toth@csun.edu}. }\\
\and  Karol W\k{e}grzycki \thanks{Max Planck Institute for Informatics. Supported by the Deutsche Forschungsgemeinschaft (DFG, German Research Foundation) grant number 559177164. Email:     \texttt{kwegrzyc@mpi-inf.mpg.de}. }\\
\and Sampson Wong \thanks{University of Copenhagen. Supported by the European Union's Marie Skłodowska-Curie Actions Postdoctoral Fellowship, grant number 101146276. Email: \texttt{sampson.wong123@gmail.com}.}}
\date{\vspace{-3\baselineskip}}

\begin{document}

\maketitle
\begin{abstract}
A Euclidean noncrossing Steiner $(1+\eps)$-spanner for a point set $P\subset\mathbb{R}^2$ is a planar straight-line graph that, for any two points $a, b \in P$, contains a path whose length is at most $1+\eps$ times the Euclidean distance between $a$ and $b$. 
We construct a Euclidean noncrossing Steiner $(1+\eps)$-spanner with $O(n/\eps^{3/2})$ edges for any set of $n$ points in the plane. This result improves upon the previous best upper bound of $O(n/\eps^{4})$ 
obtained nearly three decades ago.
We also establish an almost matching lower bound: There exist $n$ points in the plane for which  any Euclidean noncrossing Steiner $(1+\eps)$-spanner has $\Omega_\mu(n/\eps^{3/2-\mu})$ edges for any $\mu>0$. Our lower bound uses recent generalizations of the Szemer\'edi-Trotter theorem to disk-tube incidences in geometric measure theory. 
\end{abstract}
\newpage
\setcounter{tocdepth}{2}
\tableofcontents

\newpage

\section{Introduction}\label{sec:intro}

Spanners are a classical tool for data compression in graphs and network optimization. Formally, a $t$-spanner for an edge-weighted graph $G=(V,E,w)$ and a (stretch) parameter $t\geq 1$,  is a subgraph $H$ of $G$ in which the shortest-path distance between any two vertices in $V$ is at most $t$ times larger than in $G$~\cite{AhmedBSHJKS20}. \emph{Metric spanners} can approximate distances in a finite metric space $(X,d)$ by setting $V=X$ and the edge weights to be the metric distances between the vertices.
Geometric networks~\cite{NS07} are an important class of metric spanners, with applications in the design of physical networks in low-dimensional Euclidean spaces. Results in Euclidean spaces are also applicable for other settings via metric embeddings into Euclidean spaces~\cite{AichholzerBBCFM22,Bourgain85,LeightonM10,LLR95}. 

For Euclidean spanners in the plane, research efforts have diverged into two distinct regimes: (1) \emph{$(1+\eps)$-spanners}, where the stretch $t=1+\eps$ is arbitrarily close to 1, and the minimum size and weight of a $(1+\eps)$-spanner is bounded by a function of $1/\eps$, and (2) \emph{plane} spanners, where the edges of the spanner are noncrossing line segments in $\mathbb{R}^2$. Researchers have made strides in both regimes over the last decade, and have uncovered optimal or near-optimal trade-offs between key parameters (see details below). However, very little attention was given to spanners that meet a dual objective: $1+\eps$ stretch for arbitrarily small $\eps>0$ \emph{and} noncrossing edges in $\mathbb{R}^2$. 

This paper focuses on this dual objective. The problem may have avoided scrutiny because it has seemingly trivial answers: 
On the one hand, a simple instance of four points at the vertices of a square shows that a noncrossing straight-line Euclidean spanner cannot achieve stretch less than $\sqrt2$. If we insist on noncrossing edges and stretch close to 1, then we must allow \emph{Steiner points}.  On the other hand, if Steiner points are allowed, then the planarization of an optimal $(1+\eps)$-spanner (by introducing Steiner vertices at edge crossings) provides a \emph{noncrossing} Steiner $(1+\eps)$-spanner with the same weight. However, planarization substantially increases the number of edges (hence the size of the spanner). Our goal is to find the best trade-offs between $\eps>0$ and the number of Steiner points for a Euclidean noncrossing Steiner $(1+\eps)$-spanner in $\mathbb{R}^2$. 

\begin{problem}\label{problem:1}
   Determine $s(n,\eps)$, defined as the minimum integer such that every set of $n$ points in Euclidean plane admits a Euclidean noncrossing Steiner $(1+\eps)$-spanner with at most $s(n,\eps)$ Steiner points. 
\end{problem}
This problem was also posed as Open Problem~17 in a survey by Bose and Smid~\cite[Section~4]{BoseS13} in 2013. The best bound available at that time was given by Arikati et al.~\cite{ArikatiCCDSZ96}: They constructed a Euclidean noncrossing Steiner $(1+\eps)$-spanner with $O(n/\eps^4)$ Steiner points by taking rectangular decompositions for $n$ points in $O(1/\eps)$ equally spaced directions.

\paragraph{Previous work.}
There are several possible approaches to address \Cref{problem:1}. 
For $n$ points in the plane, there are $(1+\eps)$-spanners with $O(n/\eps)$ edges (e.g., $\theta$-graphs~\cite{Clarkson87}); and this bound is the best possible. Since the edges of a spanner may pairwise cross, a na\"ive analysis of straightforward  planarization would lead to a Euclidean Steiner $(1+\eps)$-spanner with $O(n^2/\eps^2)$ Steiner points. 

Alternatively, one could try to bound the number of edge crossings in a $(1+\eps)$-spanner (without Steiner points). Given a set $P\subset \mathbb{R}^2$ of $n$ points, the greedy $(1+\eps)$-spanner of Alth\"ofer et al.~\cite{althofer1993sparse} is constructed as follows: sort the $\binom{n}{2}$ possible edges by nondecreasing length, initialize an empty graph $H=(P,\emptyset)$, and add an edge $ab$ to $H$ if $d_H(a,b)>(1+\eps)\cdot d(a,b)$. The greedy $(1+\eps)$-spanner has $O(n/\eps)$ edges~\cite{althofer1993sparse,CDNS92,ChandraDNS95}; and Eppstein and Khodabandeh~\cite{EppsteinK21} proved that every edge $ab$ crosses $O(1/\eps^8)$ edges that are longer than $ab$. Consequently, it has $O(n/\eps^9)$ crossings: planarization would create this many Steiner points. This bound is weaker than the previous bound of $O(n/\eps^4)$ by Arikati et al.~\cite{ArikatiCCDSZ96}.

For $n$ points in the plane and $k\in \mathbb{N}$, the \emph{emanation graph of grade $k$}, introduced by Hamedmohseni et al.~\cite{HamedmohseniRM23}, is constructed by shooting $2^{k+1}$ rays from each given point, where the shorter rays stop the longer ones upon collision. The emanation graph of grade $k$ is a noncrossing Steiner spanner with $O(2^k n)$ Steiner points. Hamedmohseni et al.~\cite{HamedmohseniRM23} show that the stretch factor is at most $\sqrt{10}$, but for every $k\in \mathbb{N}$, there are point sets for which the stretch factor is arbitrarily close to $\sqrt{2}$; hence this approach does not lead to $(1+\eps)$-spanners. 

\begin{figure}[!ht]
    \centering
    \includegraphics[width=0.8\textwidth]{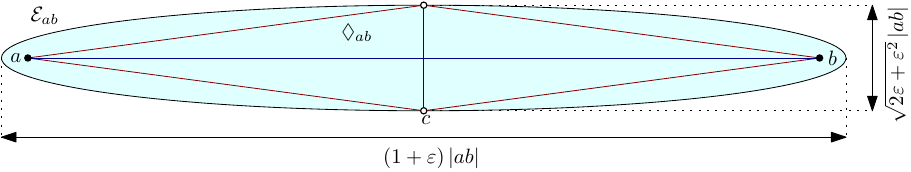}
    \caption{An ellipse $\mathcal{E}_{ab}$ with foci $a$ and $b$, and major axis of length $(1+\eps)\,|ab|$, and rhombus $\lozenge_{ab}$.
    }
    \label{fig:ellipse}
\end{figure}

\paragraph{Cone-restricted spanners.}
Let $P$ be a set of $n$ points in the plane. For any $a,b\in P$, any $ab$-path of length at most $(1+\eps)\, |ab|$ lies in an ellipse $\mathcal{E}_{a,b}$ with foci $a$ and $b$ and major axis $(1+\eps)\cdot |ab|$; see \Cref{fig:ellipse}. 
Note that for small $\eps>0$, the ellipse $\mathcal{E}_{ab}$ is long and narrow. It is known that in every $ab$-path of length at most $(1+\eps)\, |ab|$, the total length of the edges  
that make an angle $\alpha\leq O(\sqrt{\eps})$ with the line segment $ab$ is $\Omega(|ab|)$~\cite{BhoreT22}. The angle threshold $\alpha\leq O(\sqrt{\eps})$ is the best possible: 
For example, if $c$ is an intersection point of $\mathcal{E}_{ab}$ and its minor axis, the $ab$-path $(a,c,b)$ has length $|ac|+|cb|=(1+\eps)\, |ab|$, but both $ac$ and $cb$ make an angle of $\Theta(\sqrt{\eps})$ with $ab$. We define a variant of Euclidean $(1+\eps)$-spanners (with or without Steiner points), where we require an $ab$-path, for all $a,b\in P$, in which \emph{all} edges  make an angle $O(\sqrt{\eps})$ with the line segment $ab$. 

\begin{definition}\label{def:cone}
    Let $P$ be a set of $n$ points in $\mathbb{R}^d$, for constant dimension $d\in \mathbb{N}$, and let $\alpha\in (0,\pi/2]$.
    A (Steiner) graph $G=(V,E)$, with $P\subseteq V$, is a \EMPH{cone-restricted} (Steiner) $(1+\eps)$-spanner for $P$ if for every $a,b\in P$, there is an $ab$-path $(a=p_0,p_1,\ldots , p_m=b)$ in $G$ such that $\angle(\overrightarrow{ab},\overrightarrow{v_{i-1}v_i})\leq \sqrt{\eps}$ for all $i=1,\ldots ,m$.
\end{definition} 

Note that if $G$ is a cone-restricted (Steiner) $(1+\eps)$-spanner for $P$, then for every $a,b\in P$, there is an $ab$-path of length at most $(1+\eps)\, |ab|$ that lies in the rhombus $\lozenge_{ab}$ spanned by $a$, $b$, and the two intersection points of $\mathcal{E}_{ab}$ with its minor axis; see \Cref{lemma:cone_restricted}.

\subsection{Contributions and technical highlights}
\label{ssec:results}

\paragraph*{Upper bound.}
Our first contribution is a noncrossing Steiner $(1+\varepsilon)$-spanner with $O(n/ \varepsilon^{3/2})$ Steiner vertices. This improves upon the previous best result by Arikati et al.~\cite{ArikatiCCDSZ96}, which has $O(n/\varepsilon^4)$ Steiner vertices.

\begin{restatable}{theorem}{upperbound}

\label{thm:upper}
     For every $\eps>0$ and every set of $n$ points in Euclidean plane, there is a noncrossing Steiner $(1+\varepsilon)$-spanner with $O(n/ \varepsilon^{3/2})$ Steiner vertices. 
     Furthermore, there is such a spanner that is cone-restricted, and can be computed in $O((n \log n)/ \varepsilon^{3/2})$ time.
\end{restatable}

Arikati et al.~\cite{ArikatiCCDSZ96} construct a set of noncrossing graphs~$G_i$ for $1 \leq i \leq k$. The final spanner is $\bigcup_{i=1}^k G_i$, where a Steiner vertex is added at every edge crossing between $G_i$ and $G_j$, $i \neq j$. Our construction improves on the construction of Arikati et al.~\cite{ArikatiCCDSZ96} in several ways. First, we use fewer graphs~$G_i$ when constructing our spanner. Specifically, we use $k = O(1/\sqrt{\varepsilon})$ instead of $k = O(1/\varepsilon)$ graphs. Whereas the previous construction~\cite{ArikatiCCDSZ96} uses $O(1/\varepsilon)$ rotated copies of the point set to approximate the $L_2$ distance with the $L_1$ distance in one of the copies, we instead use $O(1/\sqrt{\varepsilon})$ carefully chosen linear transformations to achieve cone-restricted paths between all pairs of points. Second, our graphs $G_i$ are of smaller size than those in Arikati et al.~\cite{ArikatiCCDSZ96}. Both constructions obtain $G_i$ by refining the Balanced Box Decomposition into an axis-parallel spanner under the $L_1$ metric. We are able to use the properties of our linear transformations to obtain the same stretch guarantee using a coarser refinement of the Balanced Box Decomposition. In particular, we obtain $|G_i| = O(n)$, improving on $|G_i| = O(n/\varepsilon^2)$ in~\cite{ArikatiCCDSZ96}. Third, our construction fills some details that are missing from both~\cite{ArikatiCCDSZ96} and \cite[Chapter 4]{BoseS13}.

\paragraph{Lower bounds.}
Our second contribution is an almost matching lower bound on the number of
Steiner vertices in any Euclidean noncrossing Steiner $(1+\eps)$-spanner.
\begin{restatable}{theorem}{lowerbound}
\label{thm:lower}
\label{thm:lower+}
     For every sufficiently small $\eps,\mu>0$ and every $n\in \mathbb{N}$, there exists a set of $n$ points in the Euclidean plane for which every noncrossing Steiner $(1+\varepsilon)$-spanner has $\Omega_\mu(n/\varepsilon^{3/2-\mu})$ Steiner vertices, where the constant hidden in the $\Omega_\mu(.)$ notation depends only on~$\mu$.     
\end{restatable}

To the best of our knowledge, the best previous unconditional lower bound
follows from the known size lower bound of $\Omega(n/\sqrt{\eps})$ for
Euclidean Steiner $(1+\eps)$-spanners~\cite{LeS25}, which also applies to
noncrossing spanners. In this paper, the canonical example to prove lower
bounds is the following construction: consider the unit square $[0,1]^2$ and let
$A$ (resp., $B$) be the set of equally spaced points on the left (resp., right)
side of $[0,1]^2$ so that the distance between any two consecutive points in $A$
(resp., $B$) is $4\sqrt{\eps}$. The point set $P$ is the union of points in $A$
and $B$; see \Cref{fig:lb} for an illustration. Observe that: (i) the minimum
pairwise distance in $P$ is $4\sqrt{\eps}$, (ii) the maximum pairwise distance
(a.k.a., diameter) in $P$ is $\sqrt{2}$, and (iii) for any $a \in A$ and $b \in
B$, the slope of the segment $ab$ is between $-1$ and $1$. Let $G$ be a
noncrossing Steiner $(1+\eps)$-spanner with the minimal number of Steiner points
for $P$.

\begin{figure}[!ht]
    \centering
    \includegraphics[width=0.35\textwidth]{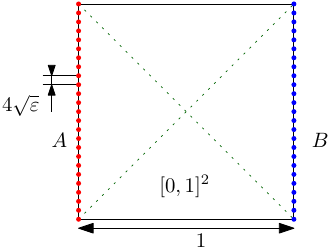}
    \caption{Point set $A\cup B$, where $A$ and $B$ lie on two opposite sides of a unit square.}
    \label{fig:lb}
\end{figure}

First, we sketch key ideas behind a weaker $\Omega(n/\eps)$ lower bound
(see~\Cref{ssec:simplebound} for the full proof).
\begin{restatable}{theorem}{lowerboundsimple}
     For every sufficiently small $\eps>0$ and every $n\in \mathbb{N}$, there exists a set of $n$ points in the Euclidean plane for which every noncrossing Steiner $(1+\varepsilon)$-spanner has $\Omega(n/\varepsilon)$ Steiner vertices.     
\label{thm:lower-}
\end{restatable}
Let $A_0 \subseteq A$ consist of every third
point from the bottom third of the left side, and $B_0 \subseteq B$ consist of
every third point from the top third of the right side. This ensures $|A_0|,
|B_0| = \Omega(1/\sqrt{\eps})$ and for each pair $(a,b)\in A_0\times B_0$ the
segment $ab$ has slope in $[\frac{1}{3}, 1]$. For each pair $(a,b)\in A_0\times
B_0$, let $\gamma_{ab}$ denote a shortest path in $G$ from $a$ to $b$. We apply
the result of Bhore and T\'oth~\cite[Lemma 4]{BhoreT22} to conclude that the set
$E_{ab}$ of edges in $\gamma_{ab}$ having angle at most $3\sqrt{\eps}$ with $ab$
satisfies $\|E_{ab}\| \geq \frac{7}{9}|ab|$, i.e., the total length of the edges in
$E_{ab}$ is at least $\frac{7}{9}|ab|$. 
Moreover, using a geometric argument, we can show that for distinct pairs $(a,b), (a',b')\in A_0\times B_0$, the sets $E_{ab}$ and
$E_{a'b'}$ are pairwise disjoint.

\begin{figure}[!ht]
    \centering
    \includegraphics[width=0.45\textwidth]{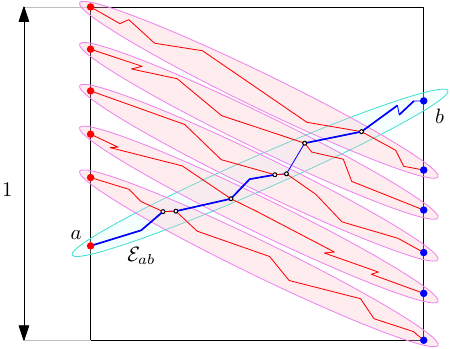}
    \caption{An $ab$-path and its intersections with other paths between point pairs of slope $-1/2$.}
    \label{fig:slices-}
\end{figure}

Now, consider the collection of all ellipses $\mathcal{E}_{cd}$ with $c,d\in A\cup B$
and slope $-\frac{1}{2}$ (see~\Cref{fig:slices-}). As these ellipses are pairwise
disjoint (\Cref{lem:ellipses}), each such ellipse creates a Steiner vertex wherever $\gamma_{ab}$ (for
$(a,b)\in A_0\times B_0$) intersects it. Since there are $\Theta(1/\sqrt{\eps})$
such disjoint ellipses, and each path $\gamma_{ab}$ must cross
$\Theta(1/\sqrt{\eps})$ of them, each path $\gamma_{ab}$ contains
$\Theta(1/\sqrt{\eps})$ Steiner points in these intersections. The bound on the
total length $|\gamma_{ab}| \leq (1+\eps)|ab| < 2$ combined with the pigeonhole
principle implies that the average length of an edge in $E_{ab}$ is 
$O(\sqrt{\eps})$. Consequently, $E_{ab}$ contains at least
$\Omega(1/\sqrt{\eps})$ edges. Since there are $\Omega(1/\eps)$ pairs in
$A_0\times B_0$ and the sets $E_{ab}$ are edge-disjoint, the total
number of edges is $\Omega(1/\eps) \cdot \Omega(1/\sqrt{\eps}) =
\Omega(1/\eps^{3/2}) = \Omega(n/\eps)$ for the basic construction. See
\Cref{ssec:simplebound} for a detailed proof.

This lower bound is already stronger than
the lower bound of $\Omega(n/\sqrt{\eps})$ that can be derived from the crossing case~\cite{LeS25}. However, the limitation of this approach is the fact that the pigeonhole argument can give only $\Omega(1/\sqrt{\eps})$ edges per path
$\gamma_{ab}$. To improve this to $\Omega_\mu(n/\eps^{3/2-\mu})$ for any $\mu >
0$, we will use techniques from geometric measure theory
(see~\Cref{ssec:improved} for the full proof of~\Cref{thm:lower+}). To use these tools, we first need to
show that most of the spanner paths have certain properties. Let $M$ be the
square of side-length $1/8$ in the middle of $[0,1]^2$. We decompose $M$ into an
$O(1/\sqrt{\eps}) \times O(1/\sqrt{\eps})$ grid consisting of square
\emph{windows} and analyze the structure of spanner paths within each window.
We will also restrict our attention to ellipses $\mathcal E_{ab}$ where the corresponding segment~$ab$ has slope $\lambda\in [1/4,3/4]$, and we say that such ellipses and spanner paths are in the \emph{positive bundle}, or $ab$ has slope $\lambda\in [-3/4,-1/4]$, where the corresponding ellipses and spanner paths belong to the \emph{negative bundle}.
The key properties we establish are:
\begin{itemize}
    \item Each window $W$ contains crossing ellipses for each slope in the positive and negative bundles (see \Cref{lem:lowerbasics}(i)).
    \item Each spanner path $\gamma_{ab}$ is \emph{adventurous} in at most $c/\sqrt{\eps}$ windows, meaning it goes outside a narrow strip $R_W(\gamma_{ab})$ of width 
    $O(\eps)$ in only a small number of windows (see \Cref{lem:lowerbasics}(ii)). Here $c$ is a small constant.
    \item Each spanner path $\gamma_{ab}$ is \emph{skewed} in at most $c/\sqrt{\eps}$ windows, meaning its direction deviates significantly from the direction of the segment $ab$ in only a small number of windows  (see \Cref{lem:lowerbasics}(iii)). Again, $c$ is a small constant.
    \item More than half of the windows in $M$ are \emph{well-behaved}, meaning they are crossed by sufficiently many non-adventurous, non-skewed spanner paths from both the positive and negative bundles (see \Cref{lem:wellbehaved}).
\end{itemize}
In a well-behaved window $W$, we can also identify collections $\Psi^+_W$ and $\Psi^-_W$ of spanner paths such that each path $\gamma \in \Psi^+_W \cup \Psi^-_W$ is non-adventurous, non-skewed, and the paths and their corresponding strips $R_W(\gamma)$ of width $O(\eps)$ have the following properties (see~\Cref{fig:tubes}):
\begin{itemize}
    \item Paths in $\Psi^+_W$ and $\Psi^-_W$ pairwise cross within the window (see~\Cref{lem:essenitiallydistinct}(i)).
    \item The directions of the strips $R_W(\gamma)$ corresponding to these paths differ by at least $\Omega(\sqrt{\eps})$ (see~\Cref{lem:essenitiallydistinct}(ii)).
    \item The strips corresponding to distinct paths have small intersection areas, at most half the area of any individual strip (see~\Cref{lem:essenitiallydistinct}(iii)).
\end{itemize}

\begin{figure}[!ht]
    \centering
    \includegraphics[width=0.5\textwidth]{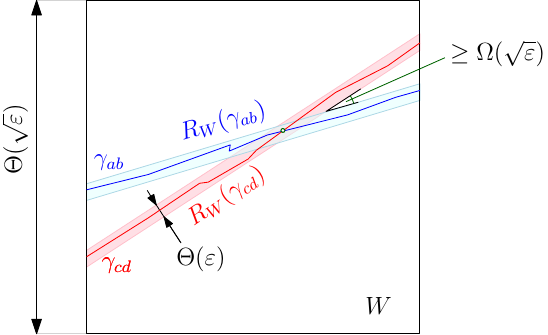}
    \caption{A well-behaved window $W$ with two non-adventurous non-skewed paths $\gamma_{ab},\gamma_{cd}\in \Psi^+_W$.}
    \label{fig:tubes}
\end{figure}

After this preparation, we want to analyze the number of Steiner vertices. A natural tool would be the classical Szemer\'edi--Trotter theorem which states that for any set of $\ell$
lines in the plane, the number of points incident to at least $r$ of these lines
is at most $O(\ell^2/r^3 + \ell/r)$. Ideally, we would like to apply this theorem to
count the number of Steiner vertices incident to many spanner paths. However,
the Szemer\'edi--Trotter theorem does not directly apply, since our
spanner paths are not straight lines.

Motivated by problems in geometric measure theory, the Szemer\'edi--Trotter theorem was generalized to \emph{disk-tube
incidences}~\cite{GSW19}. Instead of points and lines, we count intersections between $\delta$-disks (small disks of radius $\delta$) and $\delta$-tubes (long rectangles of width $\delta$). Fu, Gan and Ren~\cite{FGU22} recently proved that
only $O_\mu(|\mathcal{T}|^2/r^3)$ disjoint $r$-rich disks (each intersecting at least $r$ tubes) exist when tubes in $\mathcal{T}$ are sufficiently
well-spaced and separated in direction.\footnote{The exact definition of
\emph{sufficiently well-spaced} is quite technical and is guided by the $\mu >
0$ parameter (see~\Cref{thm:FGR22} for details).}

We apply their theorem to the tubes formed by the strips $R_W(\gamma_{ab})$ corresponding to the spanner
paths with $\delta = \Theta(\sqrt{\eps})$ for every well-behaved window $W$ (see
\Cref{fig:tubes}).  Properties (i), (ii) and (iii) of
\Cref{lem:essenitiallydistinct} guarantee that the collection of tubes in each
well-behaved window $W$ satisfies the spacing and direction separation
conditions required by the result of Fu, Gan and Ren~\cite{FGU22}.  

The theorem of Fu, Gan and Ren~\cite{FGU22} implies that only $O_\mu(1/(\eps
r^3))$ disjoint $r$-rich $\delta$-disks can intersect tubes in our collection.
In their result, there is a technical requirement that $r > \delta^{1-2\mu}
|\mathcal{T}_{\max}|$ (see~\Cref{thm:FGR22} for details), and so the smallest value of $r$ we can choose is $r = \Omega_\mu(1/\eps^\mu)$. This means that most of the
$\Theta(1/\eps)$ crossing points are not covered by the $r$-rich disks.
Therefore, there must exist $\Omega_\mu(1/\eps)$ crossing points where the
Steiner vertex is incident to fewer than $r_0 = \Theta_\mu(1/\eps^\mu)$ spanner
paths from our tube collection. Counting the tube crossings at these low-degree
Steiner vertices yields at least $\Omega_\mu(1/\eps) / r_0^2 =
\Omega_\mu(1/\eps^{1-2\mu})$ such vertices per window, and summing over all
$\Omega(1/\eps)$ well-behaved windows gives the $\Omega_\mu(1/\eps^{2-2\mu})$
lower bound on the number of Steiner vertices in $G$. Since $n =
\Theta(1/\sqrt{\eps})$, this implies the desired lower bound of
$\Omega_\mu(n/\eps^{3/2-2\mu})$. Finally, we scale~$\mu$ by a factor of 2. This concludes the sketch of the proof of \Cref{thm:lower}, see~\Cref{ssec:improved} for details. 

\Cref{thm:lower} leaves a small gap of roughly $\eps^{o(1)}$ between our upper and lower
bounds. We are able to show that this gap can be closed if we restrict our
attention to cone-restricted spanners which were used in the construction of
\Cref{thm:upper}. The proof of Theorem~\ref{thm:cone} (in \Cref{sec:cone}) uses a generalization of the celebrated Crossing Lemma to so-called degenerate crossings~\cite{AckermanP13,PachT09}.

\begin{restatable}{theorem}{conerestrictedthm}
\label{thm:cone}
     For every sufficiently small $\eps>0$ and every $n\in \mathbb{N}$, there exists a set of $n$ points in the plane for which every cone-restricted plane Steiner $(1+\varepsilon)$-spanner has $\Omega(n/\varepsilon^{3/2})$ Steiner vertices. Up to constant factors, this lower bound is the best possible (cf.~\Cref{thm:upper}).
\end{restatable}

\paragraph*{Organization.}
After covering further related works (\Cref{ssec:previous}) and the preliminaries (\Cref{sec:pre}), we provide the construction for our sparse noncrossing Steiner spanner (\Cref{sec:UB}). Due to space limitations, complete proofs are deferred to the appendix. The proofs of the lower bounds (\Cref{thm:lower,thm:lower-,thm:cone}) are also deferred to the appendix (\Cref{appendix:unconditional,sec:cone}).

\subsection{Further related previous work} 
\label{ssec:previous}

As noted above, Euclidean spanners have been studied under two independent regimes: 

\paragraph{$(1+\eps)$-Spanners.} 

The first Euclidean $(1+\varepsilon)$-spanners for arbitrary $\eps>0$ were obtained independently by Clarkson~\cite{Clarkson87} and Keil~\cite{keil1988approximating}, and these works also introduced the fixed-angle $\Theta$-graph (a close variant of the Yao graph~\cite{yao1982constructing}) as a basic construction tool. 
Ruppert and Seidel~\cite{ruppert1991approximating} extended this to $\mathbb{R}^d$, by giving $(1+\varepsilon)$-spanners with $O(n/\varepsilon^{d-1})$ edges for constant $d$. 
Le and Solomon~\cite{LeS25} proved that this dependence on $\varepsilon$ is tight: for every $\varepsilon>0$ and constant $d$, there exist point sets in $\mathbb{R}^d$ for which any $(1+\varepsilon)$-spanner must have $\Omega(\varepsilon^{-(d-1)})$ edges whenever $\varepsilon=\Omega(n^{-1/(d-1)})$.

Another key parameter of a spanner is its \emph{lightness}. For a set of points, it is the ratio between its total edge weight and the weight of the \textsf{MST} of the point set. 
Building on the greedy spanner of Alth\"ofer et al.~\cite{althofer1993sparse}, Das et al.~\cite{das1993optimally} proved that this construction achieves constant lightness and stretch $1+\varepsilon$ in $\mathbb{R}^3$, which was later extended to all $\mathbb{R}^d$ by Das et al.~\cite{narasimhan1995new}. 
Rao and Smith~\cite{rao1998approximating} established that the greedy $(1+\varepsilon)$-spanner in $\mathbb{R}^d$ has lightness $1/\varepsilon^{O(d)}$, and a long line of refinements culminated in the bound $O(\varepsilon^{-d}\log (1/\eps))$ by Le and Solomon~\cite{LeS25}.
Besides achieving small \emph{stretch} and \emph{sparsity}\footnote{The sparsity of a spanner is the ratio of the number of edges in the spanner to the size of MST.}, spanners are often required to satisfy additional desirable properties such as bounded degree or diameter. 
There has been extensive work on understanding the optimal trade-offs among these parameters for Euclidean spanners 
(see, e.g.,~\cite{arya1995euclidean,BhoreM25,ES15,dinitz2010low,LeMS23}).

\paragraph{Euclidean Steiner spanners.}
Steiner points can significantly reduce the weight of distance approximation structures (see, e.g.,~\cite{elkin2015steiner,Solomon15}).
Note that a Steiner $t$-spanner for a point set $P$ must guarantee stretch $t$ only for point pairs in $P$. 
Le and Solomon~\cite{LeS23} constructed Steiner $(1+\eps)$-spanners of sparsity $O(\varepsilon^{(1-d)/2})$ in $\mathbb{R}^d$ for all $d\geq 2$, and this bound is the best possible~\cite{BhoreT22}. In the plane, Bhore and T\'{o}th~\cite{BhoreT22} constructed Steiner $(1+\eps)$-spanners of lightness $O(1/\eps)$, this bound is also tight~\cite{LeS25}. 
In dimensions $d\geq 3$, the current best upper and lower bounds for lightness are $\widetilde{O}(1/\eps^{(d+1)/2})$ and $\Omega(1/\eps^{d/2})$~\cite{BhoreT22,LeS25}. 
Recent work has also analyzed online algorithms for Euclidean Steiner spanners and obtained several asymptotically tight bounds~\cite{BhoreFKT24, BhoreT25}.

\paragraph{Noncrossing spanners.}
Chew~\cite{Chew86} first proved the existence of a Euclidean spanner with stretch $\sqrt{10}$ and $O(n)$ noncrossing edges, the stretch was later improved to $2$~\cite{Chew89}. Keil and Gutwin~\cite{keil1992classes} showed that the Delaunay triangulation is a $2.42$-spanner. Later, Bonichon et al.~\cite{bonichon2012stretch} gave tight bounds of $\sqrt{4 + 2\sqrt{2}}$ for the  $L_1$- and $L_\infty$-Delaunay graphs. Subsequently, Bose et al.~\cite{bose2012pi} showed that the Yao graph is a $8\sqrt{2}$-spanner. See the comprehensive survey by Bose and Smid~\cite{BoseS13} for results on plane spanners up to 2013. Since then, several core questions raised in that survey on plane spanners have seen notable progress. The long-open existence of bounded-degree plane spanners has been resolved for degree $4$ through constructions of
Bonichon et al.~\cite{BonichonKPX15}. Later, Kanj et al.~\cite{KanjPT17} obtained improved stretch bounds. Bose et al.~\cite{bose2018improved} further tightened the trade-off by giving an algorithm to construct degree-$8$ plane spanners with stretch $\approx 4.414$. Dumitrescu and Ghosh~\cite{DumitrescuG16} strengthened lower bounds on plane-spanner dilation. 
For degree~$3$, progress has been obtained for restricted families of point sets, such as points in convex position~\cite{BiniazBCGMS17}. 
The framework has also been extended to constrained visibility, where van Renssen and Wong~\cite{RenssenW21} showed that visibility graphs among polygonal obstacles admit noncrossing spanners with bounded degree and constant stretch (see also~\cite{BhoreKK0LPT25, AbamBS19} for some recent work on plane spanners in polygonal domains). Further developments include refined stretch analysis for planar variants of $\theta$-like geometric graphs~\cite{bose2018improved}, as well as general geometric conditions under which sweepline-based constructions produce planar spanners~\cite{lee2024generalized}.

\paragraph{Previous work related to cone-restricted spanners.} 
We have defined \emph{cone-restricted} $(1+\eps)$-spanners for finite point sets in the plane (\Cref{def:cone}). Similar concepts have previously been used for other purposes.
In a \emph{geometric graph} $G=(V,E)$, the vertices are distinct points in the plane and the edges are straight-line segments. A geometric graph is \emph{strongly monotone}~\cite{AngeliniCBFP12,KindermannSSW14,FelsnerIKKMS16} if for every $a,b\in V$, there is a $ab$-path $(a=v_0,v_1,\ldots , v_m=b)$ in which $\angle (\overrightarrow{ab},\overrightarrow{v_{i-1}v_i}) \leq \pi/2$ for all $i=1,\ldots ,m$. However, this property does not guarantee any stretch factor. 

A geometric graph is \emph{angle-monotone with width $\gamma$}~\cite{BonichonBCKLV16,DehkordiFG15,LubiwM19} if for every $a,b\in V$, there is an $ab$-path in which the angle between any two edges is at most $\gamma$. For example, the axis-aligned grid graph induced by $n$ points in the plane is angle monotone with width $(\pi/2)$, and is a $\sqrt{2}$-spanner, however, it uses $O(n^2)$ Steiner vertices. Dehkordi et al.~\cite{DehkordiFG15} constructed, for $n$ points in $\mathbb{R}^2$, a plane angle-monotone graph of width $\pi/2$ using $O(n)$ Steiner points. Bonichon et al.~\cite{BonichonBCKLV16} showed that the half-$\theta_6$ graph~\cite{BonichonGHI10} is angle-monotone with width $2\pi/3$, and $O(n)$ edges, however, it is not necessarily planar. Lubiw and Mondal~\cite{LubiwM19} constructed an angle monotone graph with width $\pi/2$ with $O(n^2\log \log n/\log n)$ edges without Steiner vertices (but with crossings). They also consider a version of the problem with Steiner vertices, however, they require the angle-monotone property for all pairs of vertices (including Steiner vertices), and allow additional crossings. 
For every $\gamma>0$ and every set $P$ of $n$ points in the plane, they construct a Steiner angle-monotone graph of width $\gamma$ with $O(\frac{n}{\gamma}\log \Phi(P))$ edges, where $\Phi(P)$ is an unbounded parameter that depends on the point configurations.

Other constraints imposed on $ab$-path include \emph{greedy}~\cite{LeightonM10,PapadimitriouR04}, \emph{self-approaching}~\cite{AlamdariCGLP12} and \emph{increasing-chord}~\cite{CFG91} properties: A geometric graph $G=(V,E)$ is \emph{greedy} if for every $a,b\in V$, there is an $ab$-path  $(a=v_0,v_1,\ldots , v_m=b)$
that monotonically gets closest to $b$, that is, $d(v_i,b)< d(v_{i-1},b)$ for all $i=1,\ldots, m$. It is \emph{self-approaching} if $d(v_j,v_k)\leq d(v_i,v_k)$ for all $0\leq i<j<k\leq m$; and \emph{increasing-chord} if there is an $ab$-path that is self-approaching in both directions. The self-approaching and increasing-chord properties imply a stretch factor of 5.34~\cite{IKL99}, at most $2\pi/3$~\cite{Rote94} (see also~\cite{AichholzerAIKLR01}), resp., while the greedy property alone does not imply any stretch guarantee.

We also mention a couple of concepts that sound similar to cone-restricted spanners, but are different. The classical $\theta$- and Yao-graphs are constructed by connecting every vertex $v$ to a ``closest'' point in cones of apex $v$ and aperture $\theta$. They are known to be $O(\theta)$-spanners for $\theta\leq \pi/2$, but the $ab$-paths of length $O(\theta)\cdot d(a,b)$ are not necessarily cone-restricted: They may contain (short) edges that make an arbitrary angle with the line segment $ab$. 

Another concept, under a similar name, was introduced by Carmi and Smid~\cite{CarmiS12}: A geometric graph is \emph{$\theta$-angle-constrained} if for every vertex $v\in V$ the angle between any two edges incident to $v$ is \emph{at least} $\theta$. They note that the classical greedy $(1+\eps)$-spanner by Alth\"ofer et al.~\cite{althofer1993sparse} is $\Omega(\eps)$-angle constrained. For every $\theta \in (0,\pi/3)$, and $n$ points in the plane, one can construct a $\theta$-angle-constrained $(1+O(\theta))$-spanner in $O(n \log n)$ time~\cite{CarmiS12}; and this is not always possible for $\theta>\pi/3$~\cite{BakhsheshF17}.

\paragraph{Connections to incidences and geometric measure theory.}
Szemer\'edi and Trotter~\cite{SZT83} proved that $n$ points and $\ell$ lines in $\mathbb{R}^2$ determine $O(n^{2/3}\ell^{2/3}+n+\ell)$ point-line incidences, and this bound is the best possible~\cite{Erdos85}. Motivated by connections to geometric measure theory~\cite{Wolff99} significant progress was made on a generalization to disk-tube incidences, which is the number of intersections between well-spaced disks of radius $\delta$ and \emph{$\delta$-tubes}, where a $\delta$-tube is the $\delta$-neighborhood of a line, in a unit square $[0,1]^2$; see~\cite{CPZ25,DW25,FGU22,GSW19}. 

The number of disk-tube incidences (used in multiple scales) was instrumental in several recent breakthroughs in combinatorial geometry and geometric measure theory.  
For example, Heilbronn's classical problem~\cite{Roth51} asks for the minimum $h(n)>0$ such that any set $P$ of $n$ points in the unit square $[0,1]^2$ determines 
a nondegenerate triangle $\Delta(abc)$ of area at most $h(n)$. A line segment $ab$ is the base of a triangle $\Delta(abc)$ of area $A$ if and only if there exists a point $c\in P$ in the $\frac{2A}{d(a,b)}$-neighborhood of the line spanned by $ab$. 
Cohen et al.~\cite{cohen2023newupperboundheilbronn,CPZ25,Zha24} recently proved $h(n)\geq \Omega(n^{-8/7-1/2000})$, improving on the previous bound $h(n)\geqq \Omega(n^{8/7})$ by Koml\'os et al.~\cite{KPS82}. In another recent breakthrough  using this machinery, Ren and Wang~\cite{ren2025furstenbergsetsestimateplane} completely solved the Furstenberg set problem. 

\section{Preliminaries}
\label{sec:pre}

\paragraph{Balanced Box Decomposition.}

In our upper bound construction, we will use the Balanced Box Decomposition (BBD) of Arya et al.~\cite{AryaMNSW98}. Given a set of points~$P$, the BBD partitions the bounding box of $P$ into a set of tiles, such that each tile is either a rectangle or defined by an outer rectangle and a sticky inner rectangle. 

\begin{definition}[Sticky]
\label{definition:sticky}
In $\mathbb R^1$, an inner interval of width~$w$ is sticky with respect to an outer interval if its distances to the endpoints of the outer interval are either $=0$ or $\geq w$. In $\mathbb R^d$, an inner box is sticky with respect to an outer box if their projections onto each of the $d$ coordinate axes are sticky. 
\end{definition}

\begin{figure}[!ht]
    \centering
    \includegraphics[width=0.8\linewidth]{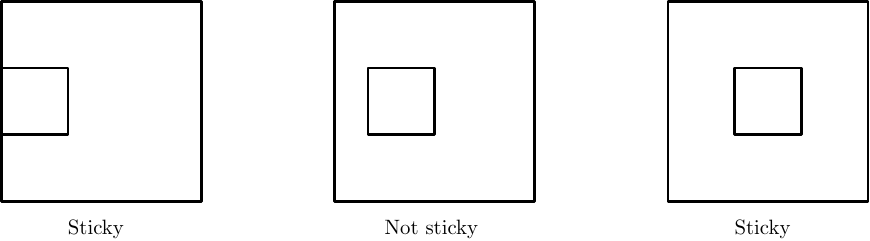}
    \caption{The middle figure shows an inner box that is not sticky, due to its $x$-projection.}
    \label{fig:sticky}
\end{figure}

\begin{theorem}[BBD~\cite{AryaMNSW98}]
\label{fact:bbd}
Given a set of points~$P$, one can partition the bounding box of $P$ into~$O(n)$ tiles such that
\begin{enumerate}[(a)]
\item each tile is either a rectangle, or 
an outer rectangle with a sticky inner rectangular hole, 
\item the rectangle, outer rectangle, and inner rectangle must have an aspect ratio of $\leq 3$,
\item each tile contains at most one point of $P$, moreover, this point lies on the tile's boundary.
\end{enumerate}
\end{theorem}

\paragraph{Properties of cone-restricted spanners.}
We prove here that every cone-restricted (Steiner) $(1+\eps)$-spanner is, in fact, a (Steiner) $(1+\eps)$-spanner; which justifies calling them $(1+\eps)$-spanners in \Cref{def:cone}.

\begin{restatable}{lemma}{conerestricted}
\label{lemma:cone_restricted}
 Let $a,b\in \mathbb{R}$, and let $\gamma=(a=p_0,p_1,\ldots, p_m=b)$ be 
 a polygonal path such that $\angle(\overrightarrow{ab},\overrightarrow{p_{i-1}p_i})\leq \sqrt{\eps}$ for all $i=1,\ldots ,m$. Then for $0<\eps\leq 1$,
 \begin{enumerate}
 \item the length of $\gamma$ is bounded by $|\gamma|\leq (1+\eps)\, |ab|$, and
 \item $\gamma \subset \lozenge_{ab}$, where $\lozenge_{ab}$ is  the rhombus $\lozenge_{ab}$ spanned by $a$, $b$, and the two intersection points of $\mathcal{E}_{ab}$ with its minor axis (cf.~\Cref{fig:ellipse}).
 \end{enumerate}
\end{restatable}
\begin{proof}
Assume, without loss of generality, that $a=(0,0)$ and $b$ lies on the positive $x$-axis. Denote the $x$- and $y$-coordinates of a point $p\in \mathbb{R}^2$ by $x(p)$ and $y(p)$, respectively. We can bound the length of $\gamma$ as follows. Using the Taylor estimate  $1-\frac{x^2}{2}\leq \cos x$, we obtain  
\begin{align*}
|\gamma| 
&=\sum_{i=1}^m |p_{i-1}p_i|
= \sum_{i=1}^m \frac{|x(p_{i-1})-x(p_i)|}{\cos\angle(\overrightarrow{ab},\overrightarrow{p_{i-1}p_i})} 
\leq \sum_{i=1}^m \frac{|x(p_{i-1})-x(p_i)|}{\cos \sqrt{\eps}}\\
&\leq \frac{ \sum_{i=1}^m x(p_i)-x(p_{i-1}) }{1-(\sqrt{\eps})^2/2}
= \frac{x(b)-x(a)}{1-\eps/2} 
= |ab|\, \sum_{k=0}^\infty \left(\frac{\eps}{2}\right)^k 
<|ab|\,(1+\eps)
\end{align*}
for $0<\eps<1$, as claimed.

For the second claim, let $c_1$ and $c_2$ be the intersection points of $\mathcal{E}_{ab}$ and the minor axis of $\mathcal{E}_{ab}$. Let $o$ be the center of $\mathcal{E}_{ab}$, which is the orthogonal projection of $c_1$ and $c_2$ onto the line spanned by $ab$. By definition, the ellipse $\mathcal{E}_{ab}$ is the locus of points $p\in \mathbb{R}^2$ such that $|ap|+|pb|=(1+\eps)\, |ab|$.
In particular, we have $|ac_1|+|c_1b|=(1+\eps)\, |ab|$.
Symmetry yields $|ac_1|=|c_1b|$, hence $|ac_1|=\frac12 (1+\eps)\, |ab|$. Consequently, $|ao|/|ac_1|=1/(1+\eps)$ and $\cos \angle(\overrightarrow{ab},\overrightarrow{ac_1})=|ac_1|/|ao|=1/((1+\eps)<1-\eps/2$. The Taylor estimate of $1-\frac{x^2}{2}\leq \cos x$ gives $\angle( \overrightarrow{ab},\overrightarrow{ac_1})<\sqrt{\eps}$.
Then for every $i=1,\ldots , m$, we have 
\begin{align*}
\tan \angle (\overrightarrow{ab},\overrightarrow{ap_i})
&=\frac{|y(p_i)-y(a)|}{|x(p_i)-x(a)|}
=\frac{\sum_{j=1}^i |y(p_{j-1})-y(p_j)|}{|x(p_i)-x(a)|}\\
&=\frac{\sum_{j=1}^i |x(p_{j-1})-x(p_j)|\tan \angle (\overrightarrow{ab},\overrightarrow{p_{j-1}p_j})}{|x(p_i)-x(a)|}\\
&\leq \frac{\sum_{j=1}^i |x(p_{j-1})-x(p_j)|\tan \sqrt{\eps}}{|x(p_i)-x(a)|}
=\frac{|x(p_i)-x(a)|\tan \sqrt{\eps}}{|x(p_i)-x(a)|}
=\tan\sqrt{\eps}
\end{align*}
This implies that $\gamma$ lies in the cone $C_a$ with apex $a$, aperture $2\sqrt{\eps}$ and symmetry axes $\overrightarrow{ab}$. Similarly, $\gamma$ lies in the cone $C_b$ of apex $b$, aperture $2\sqrt{\eps}$ and axis $\overrightarrow{ba}$.
We conclude that $\gamma\subset C_a\cap C_b = \lozenge_{ab}$.
\end{proof}

\paragraph{Basic point set for lower bounds.}
Our lower bounds (\Cref{sec:unconditional} and \Cref{sec:cone}) are based on the same basic point set, which is known to give asymptotically tight bounds for lightness and sparsity for both Steiner and non-Steiner $(1+\eps)$-spanners in the plane~\cite{BhoreT22,LeS23,LeS25}.

Let $\eps\in (0,\frac{1}{16})$, and assume w.l.o.g.\ that $\eps=4^{-k}$ for some $k\in \mathbb{N}$. We first construct a point set $P$ of size $|P|=2(2^{k-2}+1)=\Theta(\eps^{-1/2})$. 
Consider the unit square $[0,1]^2$. Let $A$ be a set of $2^{k-2}+1$ equally spaced points on the left side of $[0,1]^2$; and $B$ a set of  $2^{k-2}+1$ equally spaced points on the right side of $U$. Our point set is $P=A\cup B$; see \Cref{fig:lb}.

Note that the minimum distance between any two points in $A$ (resp., $B$) is $2^{2-k}=4\,\sqrt{\eps}$;
the diameter of $P$ is $\sqrt{2}$. Note also that for any $a\in A$ and $b\in B$, the segment $ab$ makes an angle at most $\pi/4$ with a vertical line, in particular the absolute value of the slope of $ab$ is  at most 1.

We observe two easy properties of the point set $P=A\cup B$. 
\begin{restatable}{lemma}{ellipses}
\label{lem:ellipses}
If $ab$ and $a'b'$ are parallel, then the ellipses $\mathcal{E}_{ab}$ and $\mathcal{E}_{a'b'}$ are disjoint.
\end{restatable}
\begin{proof}
First, $|aa'|=|bb'|\geq 4\,\sqrt{\eps}$ implies that the distance between the major axes of the two ellipses is more than $4\,\sqrt{\eps} / \sqrt{2} = 2\sqrt{2}\cdot \sqrt{\eps}$. By the Pythagorean theorem, the minor axis of $\mathcal{E}_{ab}$ is $\sqrt{(1+\eps)^2-1}\cdot |ab| \le \sqrt{2\eps+\eps^2} \cdot \sqrt{2} < \sqrt{6}\cdot \sqrt{\eps} < 2\sqrt{2}\cdot \sqrt{\eps}$. 
This means that $\mathcal{E}_{ab}$ (resp., $\mathcal{E}_{a'b'}$) lies in a parallel strip of width less than $2\sqrt{2}\cdot \sqrt{\eps}$ with symmetry axis $ab$ (resp., $a'b'$). Since the distance between two parallel major axes is more than $2\sqrt{2}\cdot \sqrt{\eps}$, then $\mathcal{E}_{ab}$ and $\mathcal{E}_{a'b'}$ are contained in disjoint strips, hence they are disjoint.  
\end{proof}

Next, we have a lower bound on the angle between two nonparallel segments $ab$ and $a'b'$.

\begin{restatable}{lemma}{angles}
\label{lem:angles}
For any $a,a'\in A$ and $b,b'\in B$, the following hold: 
\begin{enumerate}
    \item[(1)] if $ab$ and $a'b'$ are parallel, then $\mathcal{E}_{ab}$ and $\mathcal{E}_{a'b'}$ are disjoint;
    \item[(2)] if $ab$ and $a'b'$ are nonparallel, then $\angle(ab,a'b')> 2 \sqrt{\eps}$.
\end{enumerate} 
\end{restatable}
\begin{proof}
We can translate $a'b'$ to a segment $ab''$, where $b''$ is on the line $x=1$ that contains $B$, and $|bb''|\geq 2^{2-k}=4\,\sqrt{\eps}$. Now the law of sines for the triangle $\Delta(abb'')$ yields 
\begin{align*}
\sin \angle(ab,a'b') 
&=\sin\angle (ab, ab'') 
=\sin \angle (ab, bb'') \frac{|bb''|}{|ab''|}\\
&=\sin \angle (ab, bb'') \frac{|bb''|}{|a'b'|}
\geq \frac{\pi}{4} \cdot\frac{4\,\sqrt{\eps}}{\sqrt{2}}
\geq \frac{\pi}{\sqrt{2}}\cdot \sqrt{\eps} > 2\cdot\sqrt{\eps}.
\end{align*}
Now the inequality $\sin(x)\leq x$ 
readily gives $\angle(ab,a'b')\geq \sin(2.2 \sqrt{\eps})>2\,\sqrt{\eps}$.
\end{proof}

\section{A sparse noncrossing Steiner spanner}
\label{sec:UB}

\paragraph{Construction.}
Our construction is based on the noncrossing Steiner $(1+\varepsilon)$-spanner of Arikati et al.~\cite{ArikatiCCDSZ96}. See also Section~4 in the survey by Bose and Smid~\cite{BoseS13}.   

We are given a set $P$ of $n$ points in the plane, and a parameter $\eps>0$. 
Let $k \in \mathbb{N}$. Note that $k = k(\varepsilon)$ and will be chosen later based on $\eps$. We will construct a set of planar straight-line graphs~$G_i$ for $i\in \{1,\ldots , k\}$. Then we will construct the final spanner $G=\bigcup_{i=1}^{k} G_i$ as the union of the graphs~$G_i$, where a Steiner point is inserted at each edge crossing between edges in $G_i$ and $G_j$ for $i \neq j$. For $i\in \{1,\ldots , k\}$, the graph $G_i$ will be such that the edges have two possible directions: they either make an angle of $i\cdot \frac{\pi}{k}$ or $(i+\delta)\cdot \frac{\pi}{k}$ with the positive $x$-axis, where $\delta$ is a positive integer\footnote{We choose $\delta = 3$, whereas Arikati et al.~\cite{ArikatiCCDSZ96} choose $\delta = \frac k 2$.}. Consider the affine transformation $T_i:\mathbb{R}^2\to\mathbb{R}^2$ that maps unit vectors of direction $i\cdot \frac{\pi}{k}$ and $(i+\delta)\cdot \frac{\pi}{k}$ to unit vectors along the positive $x$- and $y$-axes, respectively. By applying the transformation on~$G_i$, we obtain an axis-parallel graph~$T_i(G_i)$ on the point set $T_i(P)$.

It remains to construct the axis-parallel graph $T_i(G_i)$. We use the Balanced Box Decomposition (BBD), which we introduced in Section~\ref{sec:pre}. Recall that the BBD divides the bounding box of $T_i(P)$ into $O(n)$ tiles. For each of these $O(n)$ axis-aligned tiles, we will further subdivide the tile into at most $9K^2$ axis-aligned rectangles. If a tile contains no hole, we subdivide it into rectangles using~$K$ equally spaced horizontal lines and~$K$ equally spaced vertical lines. If a tile contains a hole, we first subdivide it into at most 9 rectangles using the four lines spanned by the four sides of the inner rectangle, and then subdivide each of these rectangles using~$K$ equally spaced horizontal lines and~$K$ equally spaced vertical lines. See \Cref{fig:further_subdivide}.

\begin{figure}[!ht]
    \centering
    \includegraphics[width=0.7\linewidth]{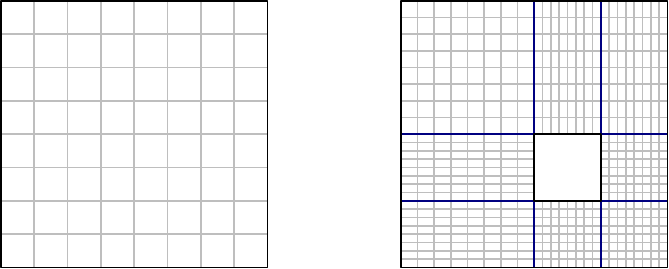}
    \caption{Further subdividing a tile into $\leq 9K^2$ rectangles, first by the four lines spanned by the sides of the inner rectangle (dark blue), and second by equally spaced axis-parallel lines (grey).}
    \label{fig:further_subdivide}
\end{figure}

This yields a partition of the bounding box of $T_i(P)$ into $O(nK^2)$ axis-aligned rectangles such that each point in $T_i(P)$ lies on the boundary of one of the rectangles. This partition defines the axis-parallel graph $T_i(G_i)$. In particular, the vertices of the graph are the points $T_i(P)$ and the vertices of the rectangles in the partition, and the edges of the graph are the edges of the rectangles in the partition, or a pair of edges connected to a vertex of $T_i(G_i)$ if the vertex lies on an edge of a rectangle. This completes the construction of~$T_i(G_i)$. We can apply the inverse transformation~$T_i^{-1}$ to obtain the graph~$G_i$. Finally, by constructing the union $G=\bigcup_{i=1}^{k} G_i$ and adding Steiner points at edge crossings, we obtain the final graph.

A key difference between our construction and that of Arikati et al.~\cite{ArikatiCCDSZ96} is that we use linear transformations, instead of rotated copies, of the BBD construction. This difference ultimately leads to our improved bound on the number of Steiner points. In particular, the properties of our linear transformations (Lemmas~\ref{lemma:strip} and~\ref{lemma:staircase}) allow us to use significantly fewer graphs and a significantly coarser refinement of the BBD, which correspond to smaller values for the parameters~$k$ and $K$, respectively. 

Next, we define the parameters $\delta$, $k$ and~$K$ and compare them to~\cite{ArikatiCCDSZ96}. Our linear transformations use $\delta = 3$, whereas the rigid motions of Arikati et al.~\cite{ArikatiCCDSZ96} correspond to $\delta = \frac k 2$. We use $k = O(\sqrt{1/\varepsilon})$ and $K = 100$, instead of the $k = O(1/\varepsilon)$,  and $K = O(1/\varepsilon)$ used in~\cite{ArikatiCCDSZ96}. Notably, we have $|G_i| = O(n)$ instead of $|G_i| = O(n/\varepsilon^2)$ in~\cite{ArikatiCCDSZ96}. We will show that, even with fewer and smaller graphs~$G_i$, we still obtain a $(1+\varepsilon)$-spanner. 

\paragraph{Stretch analysis.}

We will prove that $G$ is a $(1+ \varepsilon)$-spanner for the new construction, i.e., for the new values of~$k$, $\delta$ and $K$. Let $a , b \in P$. Define $\angle {ab}$ to be the angle between the vector~$\overrightarrow{ab}$ and the positive $x$-axis. For the remainder of this section, we will assume without loss of generality that $\angle {ab} \in [(i+1)\cdot \frac{\pi}{k}, (i+2) \cdot \frac{\pi}{k})$ where $i \in \{1, \ldots, k\}$. The next lemma is to prove that the transformation $T_i$ sends vector $\overrightarrow{ab}$ into the angle class $[\frac \pi {12}, \frac {5\pi} {12})$. See Figure~\ref{fig:angle_class}. The lemma assumes that $\varepsilon$ is sufficiently small, and thus $k$ is sufficiently large.

\begin{restatable}{lemma}{strip}
    \label{lemma:strip}
    If $\angle {ab} \in [(i+1)\cdot \frac{\pi}{k}, (i+2) \cdot \frac{\pi}{k})$, then $\angle {(T_i(a)T_i(b))} \in [\frac \pi {12}, \frac {5\pi} {12})$.
\end{restatable}

\begin{figure}[!ht]
    \centering
    \includegraphics[width=\linewidth]{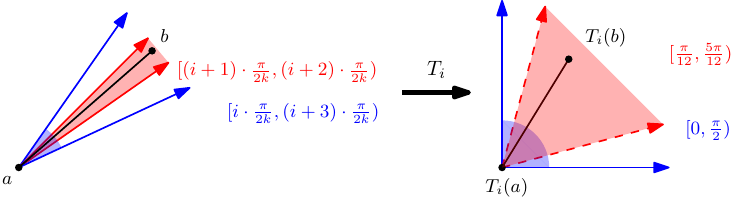}
    \caption{A visualization of the transformation~$T_i$, which sends the unit vectors in the directions $i\cdot \frac{\pi}{k}$, $(i+\delta)\cdot \frac{\pi}{k}$ (left, blue) to the unit $x$- and $y$-vectors (right, blue). Lemma~\ref{lemma:strip} states that the angle class $[(i+1)\cdot \frac{\pi}{k}, (i+2) \cdot \frac{\pi}{2k})$ (left, red) will be sent into the angle class $[\frac \pi {12}, \frac {5\pi} {12})$ (right, red).}
    \label{fig:angle_class}
\end{figure}

\begin{proof}
Recall that $T_i$ is defined to send unit vectors of direction $i\cdot \frac{\pi}{k}$ and $(i+\delta)\cdot \frac{\pi}{k}$ to unit vectors along the positive $x$- and $y$-axes, respectively. It will be much simpler to work with the inverse transformation~$T_i^{-1}$, which sends the unit $x$- and $y$-vectors to the unit vectors in the directions $i\cdot \frac{\pi}{k}$ and $(i+\delta)\cdot \frac{\pi}{k}$. In other words, 

\[
    T_i^{-1} 
        \left[ \begin{array}{c} 1 \\ 0 \end{array}\right]
     = 
        \left[\begin{array}{c} \cos(i\cdot \frac{\pi}{k}) \\ \sin(i\cdot \frac{\pi}{k}) \end{array} \right]
    , \quad
    T_i^{-1} 
        \left[ \begin{array}{c} 0 \\ 1 \end{array}\right]
     = 
        \left[\begin{array}{c} \cos((i+3)\cdot \frac{\pi}{k}) \\ \sin((i+3)\cdot \frac{\pi}{k}) \end{array}\right].
\]
So we have 
\[
    T_i^{-1} = \left[\begin{array}{c c} \cos(i\cdot \frac{\pi}{k})  & \cos((i+3)\cdot \frac{\pi}{k}) \\ \sin(i\cdot \frac{\pi}{k}) & \sin((i+3)\cdot \frac{\pi}{k}) \end{array}\right] = \left[\begin{array}{c c} \cos(i\cdot \frac{\pi}{k})  & -\sin(i\cdot \frac{\pi}{k}) \\ \sin(i\cdot \frac{\pi}{k}) & \cos(i\cdot \frac{\pi}{k}) \end{array}\right] \, \left[\begin{array}{c c} 1 & \cos \frac {3 \pi} k \\ 0 & \sin \frac {3 \pi} k \end{array}\right].
\]
Next, we consider the vector $\vec{v}$ so that $T_i(\vec{v})$ is in the direction $\frac \pi {12}$. Up to scaling, we have

\[
    v = T_i^{-1} \left[ \begin{array}{c} \cos \frac \pi {12} \\ \sin \frac \pi {12} \end{array}\right] = \left[\begin{array}{c c} \cos(i\cdot \frac{\pi}{k})  & -\sin(i\cdot \frac{\pi}{k}) \\ \sin(i\cdot \frac{\pi}{k}) & \cos(i\cdot \frac{\pi}{k}) \end{array}\right] \, \left[ \begin{array}{c} \cos \frac \pi {12} + \sin \frac \pi {12} \cos \frac {3 \pi} k \\ \sin \frac \pi {12} \sin \frac {3 \pi} k \end{array}\right].
\]
We want to show that $\vec{v}$ is in the angle class $(i\cdot \frac{\pi}{k}, (i+1) \cdot \frac{\pi}{k}]$. But $\left[\begin{smallmatrix} \cos(i\cdot \frac{\pi}{k})  & -\sin(i\cdot \frac{\pi}{k}) \\ \sin(i\cdot \frac{\pi}{k}) & \cos(i\cdot \frac{\pi}{k}) \end{smallmatrix}\right]$ is the linear transformation that rotates by an angle of $(i\cdot \frac{\pi}{k})$. So it suffices to show that $\left[ \begin{smallmatrix} \cos \frac \pi {12} + \sin \frac \pi {12} \cos \frac {3 \pi} k \\ \sin \frac \pi {12} \sin \frac {3 \pi} k \end{smallmatrix}\right]$ is in the angle class $(0, \frac \pi {k}]$. Since $\tan \theta$ is increasing for $\theta \in [0, \frac \pi 2)$, it suffices to show that

\[
0
<
\frac {\sin \frac \pi {12} \sin\frac{3\pi}{k}} {\cos \frac \pi {12} + \sin \frac \pi {12} \cos\frac{3\pi}{k}} 
\leq
\frac {\sin \frac{\pi}{k}} {\cos \frac \pi k}.
\]
The first inequality holds since all terms are positive. The second inequality follows from

\[
\begin{aligned}
&&\frac {\sin \frac \pi {12} \sin\frac{3\pi}{k}} {\cos \frac \pi {12} + \sin \frac \pi {12} \cos\frac{3\pi}{k}} 
&\leq
\frac {\sin \frac{\pi}{k}} {\cos \frac \pi k} 
\\
\iff &&
\sin \frac \pi {12} \sin \frac {3 \pi} k \cos \frac \pi k
&\leq 
\cos \frac \pi {12} \sin \frac \pi k + \sin \frac \pi {12} \cos  \frac {3 \pi} k \sin \frac \pi k
\\
\iff &&
\sin \frac \pi {12} \sin \frac {2 \pi} k
&\leq
\cos \frac \pi {12} \sin \frac \pi k \\
\iff &&
\sin \frac \pi {12} \cdot 2 \cos \frac {\pi} k 
&\leq
\cos \frac \pi {12}  \\
\impliedby &&
 \tan \frac \pi {12} 
&< \frac 1 2  \quad \text{and} \quad \cos \frac {\pi} k \leq 1 .
\end{aligned}
\]

So $T_i^{-1} \left[ \begin{smallmatrix} \cos \frac \pi {12} \\ \sin \frac \pi {12} \end{smallmatrix}\right]$ lies in the angle class $(i\cdot \frac{\pi}{k}, (i+1) \cdot \frac{\pi}{k}]$. Similarly, $T_i^{-1}  \left[ \begin{smallmatrix} \cos \frac {5\pi} {12} \\ \sin \frac {5\pi} {12} \end{smallmatrix}\right]$ lies in the angle class $((i+2)\cdot \frac{\pi}{k}, (i+3) \cdot \frac{\pi}{k}]$. Therefore, $T_i$ sends vectors in the angle class $((i+1)\cdot \frac{\pi}{k}, (i+2) \cdot \frac{\pi}{k}]$ into the angle class $[\frac \pi {12}, \frac {5\pi} {12})$, as required.
\end{proof}

Next, we use Lemma~\ref{lemma:strip} to show that there is a staircase path in the graph~$T_i(G_i)$.

\begin{restatable}{lemma}{staircase}
    \label{lemma:staircase}
    If $\angle {(T_i(a)T_i(b))} \in [\frac \pi {12}, \frac {5\pi} {12})$, then there is an $xy$-monotone axis-parallel path from $T_i(a)$ to $T_i(b)$ in the graph $T_i(G_i)$.
\end{restatable}

\begin{figure}[ht!]
    \centering
    \includegraphics[width=0.4\textwidth]{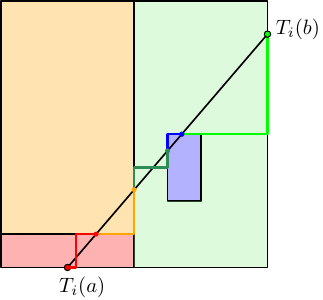}
    \caption{Replacing $T_i(a)T_i(b) \cap \tau$ with an $xy$-monotone axis-parallel path, for each tile $\tau$.}
    \label{figure:tau}
\end{figure}

\begin{proof}
    Since $a,b \in P$, we have by construction that $T_i(a)$ and $T_i(b)$ are vertices located on tile boundaries in the Balanced Box Decomposition of~$T_i(P)$. We trace the segment $T_i(a) T_i(b)$ across the tiles of the Balanced Box Decomposition, and within each tile $\tau$ we replace the segments in $T_i(a) T_i(b) \cap \tau$ with $xy$-monotone axis-parallel paths in $T_i(G_i)$. Note that $T_i(a) T_i(b) \cap \tau$ may consist of multiple segments, see the green tile in~\Cref{figure:tau}. Next, we will consider one of these segments, $cd \subseteq T_i(a) T_i(b) \cap \tau$, and perform a case analysis on the segment~$cd$.

    \begin{description}
        \item[Case 1.] Both~$c$ and~$d$ lie on the outer boundary of~$\tau$, moreover, the boundary segments containing~$c$ and~$d$ are perpendicular. The orange tile (top left) in \Cref{figure:tau} is an example. There is an $xy$-monotone axis-parallel path from~$c$ to~$d$ along the boundary of the tile~$\tau$, as required.
        
        \item[Case 2.] Both~$c$ and~$d$ lie on the outer boundary of~$\tau$, moreover, the boundary segments containing~$c$ and~$d$ are parallel. The red tile (bottom left) in~\Cref{figure:tau} is an example. Without loss of generality, assume that $c$ and $d$ lie on the top and bottom boundaries of~$\tau$, respectively. Let the height of~$\tau$, and~$cd$, be~$h$. Then the width of~$\tau$ is at most~$3h$, by \Cref{fact:bbd}. The graph $T_i(G_i)$ has $K=100$ equally spaced vertical lines, and the horizontal distance between these vertical lines is at most $\frac{3h}{100} = h \cdot 0.03$. 

        The height of the segment~$cd$ is~$h$, and by Lemma~\ref{lemma:strip}, $cd$ lies in the angle class $[\frac \pi {12}, \frac {5\pi} {12})$. Therefore, its width is at least $h / \tan \frac {5\pi} {12} = h (\frac 1 {2 + \sqrt 3}) > h \cdot 0.267$. Therefore, the width of $cd$ is greater than twice the horizontal separation between vertical lines in the graph~$T_i(G_i)$. Therefore, there is a vertical line in the $x$-span of $cd$, we can connect this vertical line to~$c$ via the top boundary and to~$d$ via the bottom boundary, to obtain an axis-parallel path from~$c$ to~$d$.
                
        \item[Case 3.] The point~$c$ lies on the inner boundary of~$\tau$ and the point~$d$ lies on the outer boundary of~$\tau$, moreover, the boundary segments containing~$c$ and~$d$ are perpendicular. The green tile (right) and the light green path (right) in \Cref{figure:tau} is an example. There is an $xy$-monotone axis-parallel path along the outer boundary segment containing~$d$ and along the line spanner by the inner boundary containing~$c$.

        \item[Case 4.] The point~$c$ lies on the inner boundary of~$\tau$ and the point~$d$ lies on the outer boundary of~$\tau$, moreover, the boundary segments containing~$c$ and~$d$ are parallel. The green tile (right) and the dark green path (middle) in \Cref{figure:tau} is an example. Without loss of generality suppose that the inner and outer boundary segments containing~$c$ and~$d$ are both vertical. The graph $T_i(G_i)$ has $K=100$ equally spaced horizontal lines with one endpoint on the inner boundary segment (that contains~$c$) and one endpoint on the outer boundary segment (that contains~$d$). The vertical distance between these horizontal lines is $\frac {h}{100} = h \cdot 0.01$.

        By the sticky property of the Balanced Box Decomposition (\Cref{definition:sticky}), the horizontal width of $cd$ is at least the width of the inner box, which is at least $\frac h 3$. By Lemma~\ref{lemma:strip}, the segment $cd$ lies in the angle class $[\frac \pi {12}, \frac {5\pi} {12})$. Therefore, the height of $cd$ is at least $\frac h 3 \tan \frac \pi {12} = h (\frac {2 - \sqrt 3} 3) > h \cdot 0.089$. The height of~$cd$ is greater than twice the vertical separation between the $K=100$ horizontal lines in the graph $T_i(G_i)$. Now, we can construct an axis-parallel path from~$c$ to~$d$ in the graph~$T_i(G_i)$. We can travel vertically from $c$ to one of these horizontal lines in $T_i(G_i)$, travel along this horizontal line, and then travel vertically to~$d$. \qedhere
    \end{description}
\end{proof}

The final step is to observe that, since there is an $xy$-monotone axis-parallel path from $T_i(a)$ to $T_i(b)$ in $T_i(G_i)$, by reversing the transformation $T_i$, we obtain a cone restricted path from $a$ to $b$ in $G$. By \Cref{lemma:cone_restricted}, the length of this path is at most $(1+\varepsilon)|ab|$. Therefore, $G$ is a $(1+\varepsilon)$-spanner. 

\paragraph{Number of Steiner points.} A Steiner vertex is inserted at each edge crossing between edges in $G_i$ and $G_j$, where $1 \leq i < j \leq k = O(\sqrt{1/\eps})$. The next lemma will help bound the number of crossings between edges in $G_i$ and $G_j$. In this lemma, we will assume that $\varepsilon$ is sufficiently small and thus $k$ is sufficiently large, so that $\sin\frac {3\pi} k = \Omega (\frac 1 k)$.

\begin{restatable}{lemma}{intersectinglongeredges}
    \label{lemma:intersecting_longer_edges}
    Let $1 \leq i <j \leq k = O(\sqrt{1/\eps})$, and~$e$ be an edge of $G_i$. Then $G_j$ contains at most $O(\sqrt{1/\eps})$ edges that both (i) intersect $e$ and (ii) are at least as long as~$e$.
\end{restatable}

\begin{proof}
    Let $\tau$ be a tile in the Balanced Box Decomposition of $T_j(P)$. Let the outer rectangle of~$\tau$ have height~$h$ and width~$w$. Recall that $T_j$ sends the unit vectors in the directions $i\cdot \frac{\pi}{k}$ and $(i+3)\cdot \frac{\pi}{k}$ to the unit $x$- and $y$-vectors. Therefore, the inverse transformation $T_j^{-1}(\tau)$ sends the axis-parallel rectangle~$\tau$ to a parallelogram. The side lengths are preserved, but the angle between the sides is not. Specifically, $T_j^{-1}(\tau)$ is a parallelogram with side lengths $h$ and $w$, and two of the four angles of the parallelogram are $\frac {3 \pi} k$. The area of the parallelogram is $hw \cdot \sin \frac {3 \pi} k$.

    Next, we will bound the number of tiles $\tau$ where $T_j^{-1}(\tau)$ intersects~$e$, and where one of the sides of $\tau$ is at least as long as~$e$. Let this set of tiles be~$J$. Each tile $\tau \in J$ has aspect ratio at most three, so both sides of~$\tau$ must be at least a third of the length of $e$. Consider the disk centered at the midpoint of~$e$, with radius twice the length of~$e$. The area of this disk is $4 \pi |e|^2$. Each parallelogram $T_j^{-1}(\tau)$, where $\tau \in J$ intersects $e$, has side lengths $\geq \frac {|e|} 3$, and has a smallest angle equal to $\sin \frac {3 \pi} k$. So each parallelogram $T_j^{-1}(\tau)$ covers a region of area $\Omega(|e|^2/k)$ in the interior of the ball, since $\sin\frac {3\pi} k = \Omega (\frac 1 k)$ for $k\geq 1$. Moreover, the parallelograms $T_j^{-1}(\tau)$ cover disjoint regions since the tiles $\tau\in J$ are disjoint. Therefore, there are $O(k)$ tiles in~$J$.

    Finally, each edge of $G_j$ that intersects $e$ and are at least as long as $e$ must lie inside some tile in~$J$. Moreover, every tile in~$J$ contains only $O(K^2) = O(1)$ edges. Therefore, since there are $O(k)$ tiles satisfying the desired property, there are also $O(k) = O(\sqrt{1/\eps})$ edges satisfying the desired property.
\end{proof}

For every edge $e$ in~$G_i$, there are $O(\sqrt{1/\eps})$ edges in $G_j$ that are longer than $e$ and cross $e$. So there are $O(1/\eps)$ edges in $\bigcup_j G_j$ longer than $e$ and crossing $e$. Next, we count the total number of edge crossings between $G_i$ and $G_j$ for $1 \leq i,j \leq k$. We charge each edge crossing to the shorter edge. There are $O(n /\sqrt{\eps})$ possible choices for the shorter edge, and $O(1/\eps)$ possible choices for the longer edge. Therefore, $\bigcup_i G_i$ has $O(n/\eps^{3/2})$ edge crossings and our noncrossing Steiner $(1+\varepsilon)$-spanner has the same number of Steiner points.

\paragraph{Running time analysis.} Computing the Balanced Box Decomposition for $n$ points in the plane takes $O(n \log n)$ time~\cite{AryaMNSW98}. Each tile can be subdivided into $O(1)$ rectangles in $O(1)$ time. So each $G_i$ can be computed in $O(n \log n)$ time. All $G_i$'s can be computed in $O((n \log n)/\sqrt{\eps})$ time. Finally, we compute the edge crossings and thus the Steiner points. All crossings among $m$ segments can be computed in $O(m\log m+s)$ time, where $s$ is the number of crossings~\cite{Balaban95,CE92}; also see~\cite[Chapter~2]{BergCKO08}. With $m=O(n/\sqrt{\eps})$ and $s=O(n/\eps^{3/2})$, the overall running time is 
bounded by $O((n\log n)/\eps^{3/2})$.

Putting this all together, we obtain the following theorem.

\upperbound*

\section{Unconditional lower bounds for noncrossing Steiner spanners}
\label{appendix:unconditional}

We start with an initial unconditional lower bound of $\Omega(n/\eps)$ (\Cref{thm:lower-} in \Cref{ssec:simplebound}) using elementary methods, and then improve it to $\Omega_\mu(n/\eps^{3/2-2\mu})$ for any constant $\mu>0$ using recent results from geometric measure theory that extend the classical Szemer\'edi--Trotter theorem to tube-disk incidences under suitable spacing conditions (\Cref{thm:lower+} in \Cref{ssec:improved}). 

\label{sec:unconditional}
\newcommand{\cE}{\mathcal{E}}
\newcommand{\cI}{\mathcal{I}}
\newcommand{\cM}{\mathcal{M}}
\newcommand{\cR}{\mathcal{R}}
\newcommand{\cW}{\mathcal{W}}
\newcommand{\cZ}{\mathcal{Z}}
\newcommand{\bd}{\partial}

\subsection{An initial bound of $\Omega(n/\eps)$}
\label{ssec:simplebound}

In \Cref{sec:cone} we proved a tight lower bound $\Omega(n/\eps^{3/2})$ in the cone-restricted setting, where all edges in a spanner path $\gamma_{ab}$ make an angle at most $\sqrt{\eps}$ with the line segment $ab$, for all  $a,b\in P$. In the unconditional setting, Bhore and T\'oth~\cite{BhoreT22} proved that at least a constant fraction of the length of $\gamma_{ab}$ consists of edges that make angle $O(\sqrt{\eps})$ with the line segment $ab$. For a set $E$ of edges in a geometric graph, let $\|E\|=\sum_{uv\in E} |uv|$ denote the total length of the edges in $E$.

\begin{lemma}[Bhore and T\'oth~\cite{BhoreT22}]\label{lem:BT22}
Let $\gamma_{ab}=(a=v_0,v_1,\ldots, v_m=b$ be a polygonal $ab$-path of length at most $(1+\eps)\,|ab|$ between points $a, b\in \mathbb{R}^d$. For an angle $\alpha\in [0,\pi/2)$, denote by $E_{ab}(\alpha)$ the set of edges $e$ in $\gamma_{ab}$ with $\angle (\overrightarrow{ab},\overrightarrow{v_{i-1}v_i}) <\alpha$. 
Then for every $i\in \{1,\ldots ,\lfloor\frac{\pi}{2}/\sqrt{\eps}\rfloor\}$, 
we have $\|E_{ab}(i\cdot \sqrt{\eps})\|\geq (1-2/i^2)\, |ab|$.
\end{lemma}
In particular, for $i=3$, we obtain $\|E(3\cdot \sqrt{\eps})\|\geq \frac79\, |ab|$.
We can now prove the main result of this section.
\lowerboundsimple*
\begin{proof}
First, we prove the lower bound for the basic construction of $n=\Theta(\sqrt{1/\eps})$ points $P=A\cup B$ in \Cref{sec:pre}. Recall that it consists of equally spaced points $A$ and $B$ on the left and right sides of a unit square $[0,1]^2$. Specifically, the distance between consecutive points in $A$ (resp., $B$) is $4\sqrt{\eps}$. 

Let $G$ be a noncrossing Steiner $(1+\eps)$-spanner for $P$.  For every point pair $(a,b)\in A\times B$, let $\gamma_{ab}$ be a shortest $ab$-path in $G$, which has length at most $(1+\eps)\, |ab|$. Since we use \Cref{lem:BT22} with $i=2$, we use the shorthand notation $E_{ab}=E_{ab}(3\cdot \sqrt{\eps})$.

We choose a subset $A_0\times B_0\subset A\times B$ such that for all point pairs $(a,b)\in A_0\times B_0$, the line segment $ab$ has nonnegative slope, and the sets $E_{ab}$ are pairwise disjoint.
Let $A_0=\{ (0,12i\cdot \sqrt{\eps}): 0\leq 12i\cdot \sqrt{\eps}\leq \frac13\}$, that is, we choose every third point in $A$, from the bottom one-third of the left side of $[0,1]^2$. Similarly, let $B_0=\{ (1,12i\cdot \sqrt{\eps}): \frac23\leq 12i\cdot \sqrt{\eps}\leq 1\}$, that is, every third point in $B$, from the top one-third of the right side of $[0,1]^2$ to the upper-right corner of $[0,1]^2$. It is clear from the construction that $|A_0|\geq \frac13\, |A|\geq \Omega(\sqrt{1/\eps})$, $|B_0|\geq \frac13\, |B| = \Omega(\sqrt{1/\eps})$, and for all $(a,b)\in A_0\times B_0$, we have $\mathrm{slope}(ab)\geq \frac13$.

We claim that for any two distinct point pairs $(a,b),(a',b')\in A_0\times B_0$, we have $E_{ab}\cap E_{a'b'}=\emptyset$. Indeed, we can argue as in the proof of \Cref{lem:angles}:
Translate $a'b'$ to a segment $ab''$, where $b''$ is on the line $x=1$, and $|bb''|\geq 8\,\sqrt{\eps}$. The law of sines for the triangle $\Delta(abb'')$ yields 
\begin{align*}
\sin \angle(ab,a'b') 
&=\sin\angle (ab, ab'') 
=\sin \angle (ab, bb'') \frac{|bb''|}{|ab''|}\\
&=\sin \angle (ab, bb'') \frac{|bb''|}{|a'b'|}
\geq \frac{\pi}{4} \cdot\frac{12\,\sqrt{\eps}}{\sqrt{2}}
\geq \frac{3\pi}{\sqrt{2}}\cdot \sqrt{\eps} >6\cdot\sqrt{\eps}.
\end{align*}
The inequality $\sin(x)\leq x$ readily gives $\angle(ab,a'b')> 6\,\sqrt{\eps}$. For any two edges, $u_{i-1}u_i\in E_{ab}$ and $v_{j-1}v_j\in E_{a'b'}$, the triangle inequality yields 
$\angle (\overrightarrow{u_{i-1}u_i}, \overrightarrow{v_{j-1}v_j}) 
\geq \angle(\overrightarrow{ab},\overrightarrow{a'b'})-\angle(\overrightarrow{ab},\overrightarrow{u_{i-1}u_i})-\angle(\overrightarrow{ab},\overrightarrow{v_{j-1}v_j})
>(6-3-3)\cdot \sqrt{\eps}>0$. This completes the proof of the claim.

\begin{figure}[!ht]
    \centering
    \includegraphics[width=0.8\textwidth]{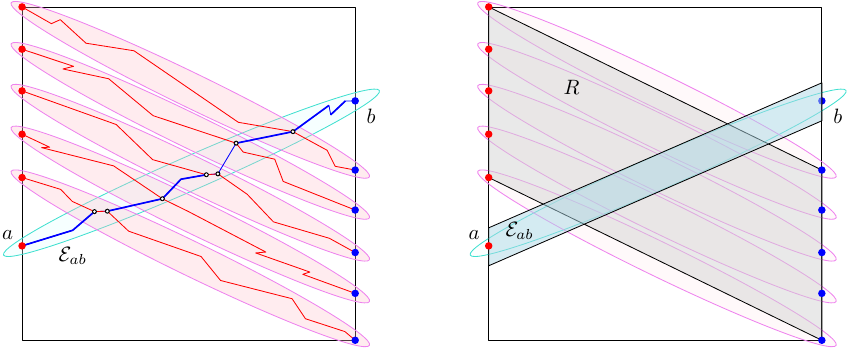}
    \caption{Left: An $ab$-path and its intersections with other paths between point pairs of slope $-1/2$. Right: Parallelogram $R$ and a strip containing the ellipse $\mathcal{E}_{ab}$.}
    \label{fig:slices}
\end{figure}

Consider the parallelogram $R$ spanned by the top half of the left edge of $[0,1]$ and the bottom half of the right edge of $[0,1]$; see \Cref{fig:slices}(right). For all $(a,b)\in A_0\times B_0$, the path $\gamma_{ab}$ crosses $R$ between the two sides of of slope $-\frac12$, and the distance between these sides is $\frac{\sqrt{3}}{4}$. 
The length of $\gamma_{ab}\cap R$ is at least $\sqrt{3}/4$. 
On the one hand, $|\gamma_{ab}\setminus R|\leq (1+\eps)\, |ab|-\frac{\sqrt{3}}{4}$.
On the other hand, $\|E_{ab}\|\geq \frac79\, |ab|$. For the set $E_{ab}\cap R=\{e\cap R: e\in E_{ab}\}$, we have 
\begin{align*}
\|E_{ab}\cap R\| 
&\geq \|E_{ab}\| - |\gamma_{ab}\setminus R|
\geq \frac79\, |ab| -\left((1+\eps)\, |ab|-\frac{\sqrt{3}}{4}\right)
\geq \frac{\sqrt{3}}{4} - \left(\frac29+\eps\right)\, |ab|\\
&\geq \frac{\sqrt{3}}{4} - \left(\frac29+\eps\right)\, \sqrt{2}
\geq \frac{\sqrt{3}}{4} -\frac13> \frac{1}{11} =\Omega(1)
\end{align*}
if $\eps>0$ is sufficiently small.

Recall that $|A_0|=|B_0|=\Omega(\sqrt{1/\eps})$, and so there are $\Omega(1/\eps)$ pairs $(a,b)\in A_0\times B_0$. 
For each $(a,b)\in A_0\times B_0$, the path $\gamma_{ab}$ crosses $\Theta(\sqrt{1/\eps})$ disjoint/parallel ellipses $\mathcal{E}_{cd}$ of slope $-\frac12$; see \Cref{fig:slices}(left). The path $\gamma_{ab}$ contains a Steiner point in the intersection with each such ellipse. 
Consequently, the length of each edge in $E_{ab}\cap R$ is at most $O(\sqrt{\eps})$. By the pigeonhole principle, $E_{ab}\cap R$ contains $\Omega(\sqrt{1/\eps})$ edges for every pair $(a,b)\in A_0\times B_0$. Overall, $\bigcup_{(a,b)\in A_0\times B_0} E_{ab}$ contains 
at least $\Omega(1/\eps^{3/2})=\Omega(n/\eps)$ edges.  

In general, for given $\eps>0$ and $n$, we use $\Theta(n/\sqrt{\eps})$ disjoint copies of the basic construction above, and obtain a lower bound $\Theta(n/\sqrt{\eps})\cdot \Omega(1/\eps^{3/2})= \Omega(n/\eps)$, as required.
\end{proof}

\subsection{An improved lower bound}
\label{ssec:improved}

For a given $\eps>0$, we use the basic construction $P=P(\eps)$ with the point sets $A$ and $B$ on the left and right of the unit square $[0,1]^2$ defined in \Cref{sec:pre}.
For a pair $(a,b)\in A\times B$, we the slope of the line segment $ab$ is $\slope(ab)=y(b)-y(a)$. Notice that $-1\leq \slope(ab)\leq 1$ for every $(a,b)\in A\times B$. Let $\mathrm{Slo}$ denote the set of possible slopes. Let $M:=[\frac{7}{16},\frac{9}{16}]^2$ be a small square in the interior of $[0,1]^2$; in this proof, we will be concerned with the intersections between spanner paths in $M$.

We say that an ellipse $\cE_{ab}$ \EMPH{crosses} an axis-aligned rectangle (or segment) $S$ if and only if both $\cE_{ab}\setminus S$ and $S\setminus \cE_{ab}$ are disconnected. The \EMPH{positive bundle} consists of the pairs $(a,b)\in A\times B$ where
$\slope(ab)\in [\frac14,\frac34]$ and
$\cE_{ab}$ crosses the middle square $M=[\frac{7}{16},\frac{9}{16}]^2$.
The negative bundle is defined analogously with target slope in $[-\frac34,-\frac14]$. Recall that we have $|A|=|B|=1/(4\sqrt{\eps})$, and they are equidistant on the left and right of $[0,1]^2$, thus there are $|\mathrm{Slo}|=1/(2\sqrt{\eps})$ slopes in total. In particular, both the positive and negative bundle has $1/(8\sqrt{\eps})$ slopes.

Let us now decompose $M$ into a $2^{-8}/\sqrt{\eps}\times 2^{-8}/\sqrt{\eps}$ grid consisting of square \EMPH{windows} of side length $32\sqrt{\eps}$; let $\cW$ be the set of windows in $M$. The \EMPH{slit} of a window $W$ is the horizontal segment of length $32\sqrt{\eps}$ connecting the midpoints of the left and right sides of~$W$, see \Cref{fig:windowslit}. Observe that there are $|\cW|=2^{-16}/\eps$ windows in $M$.

Let $G$ be a Euclidean noncrossing Steiner $(1+\eps)$-spanner for $A\cup B$, and let $\gamma_{ab}$ be a shortest path connecting $a\in A$ to $b\in B$ in $G$. Let $W$ be some window in $\cW$. Let $a',b'\in \bd Z\cap \gamma_{ab}$ be the first entry and last exit of $\gamma_{ab}$ to/from $W$, and suppose that $|a'b'|\geq 16\sqrt{\eps}$, i.e., at least half the side length of $W$. We call such a spanner path a \EMPH{core path} of $W$, see \Cref{fig:windowslit}(ii). Moreover, let $R_W(\gamma_{ab})$ denote the intersection of $W$ and the strip of width $2^{k}\eps$ centered on the line $a'b'$, where we will set $k=15$. We call $R_W(\gamma_{ab})$ the \EMPH{strip} of $\gamma_{ab}$ with respect to the window $W$. We say that $\gamma_{ab}$ is \EMPH{adventurous} in $W$ if $\gamma_{ab}\cap Z$ is not contained in $R_W(\gamma_{ab})$. We say that $\gamma_{ab}$ is \EMPH{skewed} in $W$ if the slope of the line $a'b'$ and $\slope(ab)$ differs by at least $2^{k}\sqrt{\eps}$.

\begin{figure}[t]
    \centering
    \includegraphics[scale=0.85]{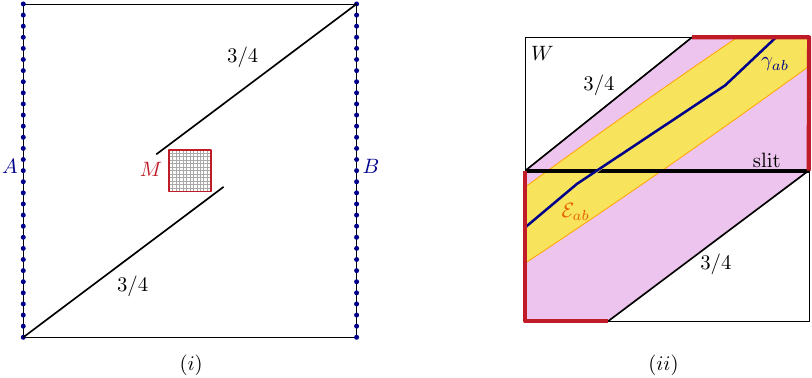}
    \caption{\textit{(i)} The base square $[0,1]^2$ with the middle square $M$, which is decomposed into windows. Any line intersecting $M$ of slope $\lambda$ where $|\lambda|<3/4$ also intersects the left and right side of $[0,1]^2$.  \textit{(ii)} A window $W$ with its slit. The path $\gamma_{ab}$ whose target slope is from the positive bundle stays inside its ellipse $\cE_{ab}$. The ellipse crosses the window slit. }
    \label{fig:windowslit}
\end{figure}

We will now make some geometric observations about the interaction of the windows, the ellipses $\cE_{ab}$, the spanner paths $\gamma_{ab}$, and the regions $R_W(\gamma_{ab})$.

\begin{lemma}\label{lem:lowerbasics}
    For every $1\leq t\leq 2k-11$, there exists a threshold $\eps_0>0$ such that for all $\eps\in (0,\eps_0)$ and set $P\subset A\cup B$, every Euclidean noncrossing Steiner $(1+\eps)$-spanner for $P$ has the following properties.
    \begin{description}
        \item[(i)] For each window $W$ of $\cW$ and each slope $\lambda \in \mathrm{Slo}$ with $|\lambda|\in [\frac14,\frac34]$, there exists a pair $(a,b)\in A\times B$ such that $\slope(ab)=\lambda$ and $\cE_{ab}$ crosses the slit of~$W$.
        \item[(ii)] Each spanner path $\gamma_{ab}$ is adventurous in at most $2^{-t}/\sqrt{\eps}$ windows.
        \item[(iii)] Each spanner path $\gamma_{ab}$ is skewed in at most $2^{-t}/\sqrt{\eps}$ windows.
    \end{description}
\end{lemma}

\begin{proof}
    (i) Consider the parallel segments $ab$ with $a\in A$ and $b\in B$ for a fixed slope $|\lambda|\in [\frac14,\frac34]$.  Recall that the distance between consecutive points of $A$ is $4\sqrt{\eps}$, so $|\lambda|\geq \frac14$ implies that any horizontal segment of length $4\cdot 4\sqrt{\eps}=16\sqrt{\eps}$ inside $M$ intersects at least one segment of slope~$\lambda$. Consider now a horizontal segment $s$ of length $32\sqrt{\eps}$, The center point $c_s$ of $s$ is sandwiched between consecutive segments $a_1b_1$ and $a_2b_2$ of slope~$\lambda$. Assume without loss of generality that $p_1:=s\cap a_1b_1$ is closer to $c_s$ than $p_2:=s\cap a_2b_2$. Consequently, the distance of $p_1$ to the center of $s$ is at most $8\sqrt{\eps}$.
    
    We claim that $\cE_{a_1b_1}$ crosses $s$. Let $a_0b_0$ be the segment of slope $\lambda$ such that $a_1$ and $b_1$ are the midpoints of $a_0a_2$ and $b_0b_2$, respectively. Let $p_0$ denote the intersection of the line of $s$ and $a_0b_0$. Then $|p_0p_1|\leq 16\sqrt{\eps}$ by our earlier argument, and the midpoint $p_{0.5}$ of $p_0p_1$ must lie on $s$. Similarly, the midpoint $p_{1.5}$ of $p_1p_2$ also lies on $s$.
    
    Since the ellipses $\cE_{ab}$ corresponding to parallel segments are pairwise disjoint by \Cref{lem:ellipses} and contain the corresponding segment $ab$, then $s\cap \cE_{a_1b_1}$ is sandwiched between $p_{0.5}$ and $p_{1.5}$, thus $\cE_{a_1b_1}$ crosses $s$. We can apply this observation on the slit of~$W$ to conclude the proof of~(i).
    \medskip

    (ii) Denote the first entry and last exit of $\gamma_{ab}$ to/from the window $W$ by $a',b'\in \bd W$, respectively. We claim that if $\gamma_{ab}$ is adventurous in $W$, then the part of $\gamma_{ab}$ between $a'$ and $b'$ (denoted by $\gamma_{ab}[a',b']$) has length at least $|a'b'|+2^{2k-10}\eps^{3/2}$. Indeed, there must be some point $p\in \gamma_{ab}$ between $a'$ and $b'$ that is outside $R_W(\gamma_{ab})$. Let $q$ denote the perpendicular projection of $p$ to the line $a'b'$, and consider the right triangles $a'qp$ and $pqb'$. Observe that $|pq|\geq 2^{k-1}\eps$ and $|a'q|+|qb'|\geq |a'b'|$. Assume without loss of generality that $|a'q|\leq |qb'|$.
    
    Suppose, for the sake of contradiction, that $|qb'|<|a'b'|+2^{2k-10}\eps^{3/2}$. This implies
    \begin{equation}\label{qb'}
        |qb'| < \mathrm{diam}(W)+2^{2k-10}\eps^{3/2} = 32\sqrt{2\eps}+2^{2k-10}\eps^{3/2} < 64\sqrt{\eps},
    \end{equation}
    where the last step uses $\eps < 2^{10-2k}$. By the triangle inequality and the Pythagorean theorem, we can bound the length of $\gamma_{ab}[a',b']$ as follows:
    \begin{align*}
        |\gamma_{ab}[a',b']| &\geq |a'p|+|pb'| = \sqrt{|a'q|^2+|pq|^2} + \sqrt{|qb'|^2+|pq|^2}\\
        &>  |a'q| + \sqrt{|qb'|^2+2^{2k-2}\eps^2}\\
        &> |a'q|+|qb'|+\frac{2^{2k-2}\eps^2}{4|qb'|}\\
        &> |a'b'|+2^{2k-10}\eps^{3/2}.
    \end{align*}
    Here the third inequality uses that $\sqrt{x^2+y^2}\geq x +\frac{y^2}{4x}$ when $y<x$, with substitution $y=2^{k-1}\eps$ and $x=|qb'|$, and the last inequality uses $|qb'|<64\sqrt{\eps}$ from \eqref{qb'}.

    If $\gamma_{ab}$ is adventurous in some set $\cW'\subset \cW$ of at least $2^{-t}/\sqrt{\eps}$ windows, then the length of $\gamma_{ab}$ is at least
    \[|ab|+\sum_{W\in \cW'} 2^{2k-10}\eps^{3/2}\geq |ab|+2^{2k-10-t}\eps\geq|ab|+2\eps > (1+\eps)|ab|\]
    as $|ab|<2$, which contradicts the fact that $\gamma_{ab}$ is a shortest $ab$-path in a $(1+\eps)$-spanner. 
    \medskip

    (iii) We will show that if $\gamma_{ab}$ is skewed in $W$, then $\gamma_{ab}[a'b']$ is significantly longer than its projection onto the segment $ab$, and we will use this to lower bound the length of any path $\gamma_{ab}$ that is skewed in many boxes. Let $\psi(\cdot)$ be the orthogonal projection to the line $ab$. We will show that $|a'b'|\geq |\psi(a')\psi(b')|+ 2^{2k+2}\eps^{3/2}$.
    First, observe that the length of $\gamma_{ab}$ between the points $a'$ and $b'$ is at least $|a'b'|$. Recall that by the core path property, $|a'b'|\geq 16\sqrt{\eps}$. We distinguish two cases.
    \begin{description}
        \item[Case 1.]  $|\psi(a')\psi(b')|< |a'b'|/2$\\
        Then we get:
        \[|a'b'|> |\psi(a')\psi(b')|+|a'b'|/2 \geq |\psi(a')\psi(b')| + 8\sqrt{\eps} \geq |\psi(a')\psi(b')| + 2^{2k+2}\eps^{3/2},\]
        where the last step uses $\eps<2^{1-2k}$.
        \item[Case 2.] $|\psi(a')\psi(b')|\geq |a'b'|/2$\\
    Let $\alpha$ be the angle between $a'b'$ and $ab$. The length of the projection $|\psi(a')\psi(b')|$ satisfies
        \[|a'b'|= \frac{|\psi(a')\psi(b')|}{\cos(\alpha)} \geq (1+\alpha^2/2)|\psi(a')\psi(b')|\]
    by the Taylor estimate of $\sec(x)=1/\cos(x)$ around $0$ as well as $\alpha<\pi/2$. Since $\gamma_{ab}$ is skewed, we have $\alpha \ge  2^{k}\sqrt{\eps}$, we get
    \[|a'b'|\geq (1+2^{2k-1}\eps)|\psi(a')\psi(b')|\geq |\psi(a')\psi(b')|+ 2^{2k+2}\eps^{3/2},\]
    where the last step uses $|\psi(a')\psi(b')|\geq |a'b'|/2\geq 8\sqrt{\eps}$.
    \end{description}

    Suppose now that $\gamma_{ab}$ is skewed in some set $\cW'\subset W$ of at least $2^{-t}/\sqrt{\eps}$ windows.
    Then its length is at least
    \[
    |\gamma_{ab}|
    \geq |ab|+\sum_{W\in \cW'} 2^{2k+2}\eps^{3/2}
    \geq |ab|+2^{2k+2-t}\eps
    >|ab|+2\eps 
    > (1+\eps)\,|ab|,
    \]
    where the last inequality used $|ab|<2$. This contradicts our assumption that $\gamma_{ab}$ is a shortest path in a $(1+\eps)$-spanner.
\end{proof}

We say that $W$ is \EMPH{well-behaved} if in both the positive and negative bundle of $W$  at least $2^{-4}/\sqrt{\eps}$ ellipses $\cE_{ab}$ cross the slit of $W$ such that the corresponding spanner paths $\gamma_{ab}$ are neither adventurous nor skewed.

\begin{lemma}\label{lem:wellbehaved}
    There are at least $|\cW|/2$ well-behaved windows in $\cW$.
\end{lemma}

\begin{proof}
    Suppose for the sake of contradiction that less than $|\cW|/2$ windows are well-behaved. In each non-well-behaved window $W$ there are less than $2^{-4}/\sqrt{\eps}$ crossing ellipses from the positive or negative bundle that are neither adventurous nor skewed. Without loss of generality, assume that there are at least $|\cW|/4$ windows where the source of error is the positive bundle, and at least $|\cW|/8$ boxes where the source of error is that there are less than $2^{-4}/\sqrt{\eps}$ crossing ellipses that are not adventurous. (The negative/skewed cases can be handled analogously.) 

    Suppose that $W$ is such a bad window, i.e., that $W$ has less than $2^{-4}/\sqrt{\eps}$ non-adventurous crossing ellipses from the positive bundle.  By \Cref{lem:lowerbasics}(i) the window $W$ has at least one crossing ellipse for each slope in $[\frac14,\frac34]$. Since there are at least $2^{-3}/\sqrt{\eps}$ slopes in the positive bundle, we conclude that there are at least $2^{-4}/\sqrt{\eps}$ slopes $\lambda\in \mathrm{Slo}$ among these such that the crossing ellipse of $W$ of slope $\lambda$ is adventurous. Thus, there are at least 
    \[\frac{|\cW|}{8}\cdot \frac{2^{-4}}{\sqrt{\eps}}=\frac{2^{-16}}{2^3\eps} \cdot \frac{2^{-4}}{\sqrt{\eps}}=2^{-23}/\eps^{3/2}\]
    window--ellipse pairs of the above type.
    
    On the other hand, there are at most $|A|\cdot 2^{-3}/\sqrt{\eps}=2^{-5}/\eps$ ellipses in the positive bundle. Therefore, there is some spanner path in the positive bundle that appears at least $\frac{2^{-23}/\eps^{3/2}}{2^{-5}/\eps}=2^{-18}/\sqrt{\eps}$ times in a pair, i.e., there is some spanner path $\gamma_{ab}$ that is adventurous in at least $2^{-18}/\sqrt{\eps}$ windows. 
    By setting $t=19$ and $k=15$, this contradicts the fact that a spanner path can be adventurous at most $2^{-t}/\sqrt{\eps}$ times by \Cref{lem:lowerbasics}(ii).
\end{proof}

\paragraph{Disk-tube incidences.}
The classical theorem by Szemer\'edi and Trotter~\cite{SZT83} shows that for any set of $\ell$ lines in the plane the number of points incident to at least $r$ lines (i.e., \emph{$r$-rich lines}) is $O(\frac{\ell^2}{r^3}+\frac{\ell}{r})$, and this bound is the best possible. Applications in geometric measure theory motivated a generalization, where points  and lines are replaced by disks and tubes. For any $\delta>0$, a \EMPH{$\delta$-disk} is a disk of radius $\delta$, and a \EMPH{$\delta$-tube} is a $\delta\times 1$ rectangle of arbitrary orientation inside the unit square $[0,1]^2$. A $\delta$-disk is \EMPH{incident} to a $\delta$-tube if they have a nonempty intersection. Guth, Solomon and Wang~\cite{GSW19} were the first to prove an analogue of the Szemer\'edi-Trotter theorem for disk-tube incidences for well-spaced tubes. 

We use a stronger variant, by Fu, Gan, and Ren~\cite{FGU22}. Two $\delta$-disks, $B_1$ and $B_2$, are \EMPH{essentially distinct}  if $\mathrm{area}(B_1\cap B_2)\leq \frac12\,\mathrm{area}(B_1)$. Similarly, two $\delta$-tubes, $T_1$ and $T_2$, are \EMPH{essentially distinct}  if $\mathrm{area}(T_1\cap T_2)\leq \frac12\,\mathrm{area}(T_1)$.
For a set $\mathcal{T}$ of $\delta$-tubes and $r\in \mathbb{N}$, let $\mathcal{B}_r(\mathcal{T})$ be a set of essentially distinct $\delta$-disks 
that each intersect at least $r$ $\delta$-tubes in $\mathcal{T}$; the disks in $\mathcal{B}_r(\mathcal{T})$ are called \EMPH{$r$-rich}.

\begin{theorem}[Fu, Gan, and Ren~\cite{FGU22}]\label{thm:FGR22}
Let $1 \leq Y \leq X \leq 1/\delta$. 
Let $\mathcal{T}$ be a collection of essentially distinct
$\delta$-tubes in $[0,1]^2$. We also assume $\mathcal{T}$ satisfies the following spacing condition: every $\frac{1}{Y}$-tube contains at most $XY$ many tubes of $\mathcal{T}$, and the directions of these tubes are $\frac{1}{X}$-separated.
We denote $|\mathcal{T}_{\max}| := XY$ (as one can see that $\mathcal{T}$ contains $O(XY)$ tubes). Then for every $\mu>0$ and $r>\max(\delta^{1-2\mu} |\mathcal{T}_{\max}|, 1)$, the number of $r$-rich balls is bounded by
\[
|\mathcal{B}_r (\mathcal{T})| \leq O_\mu \left( |\mathcal{T}|\cdot |\mathcal{T}_{\max}| \cdot \frac{1}{r^2} \left(\frac{1}{r} +\frac{1}{Y}\right)\right),
\]
where the constant hidden in the $O_\mu(.)$ notation depends on $\mu$. 
\end{theorem}

In order to apply \Cref{thm:FGR22} we need to further restrict the positive and negative bundles.

\begin{lemma}\label{lem:essenitiallydistinct}
    In each well-behaved window $W\in \cW$ there is a collection $\Psi^+_W$ of at least $|\Psi^+_W|=\Omega(1/\sqrt{\eps})$ positive-bundle paths such that each path $\gamma\in \Psi^+_W$ is a core path in $W$ that is neither adventurous nor skewed.
    The analogous claim holds for the negative bundle that contains a collection $\Psi^-_W$. Additionally, 
    \begin{description}
        \item[(i)] Each pair of spanner paths $\gamma^+\in \Psi^+_W$ and $\gamma^-\in \Psi^-_W$ cross within $W$, that is, $\emptyset \neq R_W(\gamma)\cap R_W(\gamma')\subset W$.
        \item[(ii)] For each pair of distinct $\gamma,\gamma'\in \Psi^+_W\cup \Psi^-_W$ we have that the direction of $R_W(\gamma)$ and $R_W(\gamma')$ has angle at least $2^k\sqrt{\eps}$.
        \item[(iii)] For each pair $\gamma,\gamma'\in \Psi^+_W\cup \Psi^-_W$ we have that $\mathrm{Area}\big(R_W(\gamma)\cap R_W(\gamma')\big)<\mathrm{Area}\big(R_W(\gamma))/2$.
    \end{description}
\end{lemma}

\begin{proof}
    For a fixed well-behaved window $W$ consider the positive-bundle ellipses that (i) cross the slit of $W$ and (ii) whose corresponding spanner paths are non-skewed and non-adventurous in $W$. We claim that all such paths are core paths. Observe that any positive-bundle ellipse that crosses the slit of a mid-height window must have a slope between $1/4$ and $3/4$. Thus in particular its intersection with window $W$ must stay below the line of slope $3/4$ through the left endpoint of the slit, see \Cref{fig:windowslit}(ii). Similarly, it must stay above the line of slope $3/4$ through the right endpoint of the slit. Notice that any such path is a core path as the strip given by the above two lines intersects $\bd W$ in two connected sets whose distance is more than $8\sqrt{\eps}$, i.e., more than a quarter of the side length of~$W$.

    Let us partition the positive bundle target slopes into intervals of length $2^k \sqrt{\eps}$, i.e., $[1/4,3/4]$ is split into $\tau:=2^{-k-1}/\sqrt{\eps}$ intervals of equal length; let $\mathrm{Slo}^+_1,\dots,\mathrm{Slo}^+_\tau$ denote these slope intervals. We partition the negative bundle target slopes in a similar fashion, resulting in the intervals $\mathrm{Slo}^-_1,\dots,\mathrm{Slo}^-_\tau$. Let $\cI$ denote the set of $2^{-k}/\sqrt{\eps}$ slope intervals in both bundles together. Each slope interval has $\frac{2^k \sqrt{\eps}}{4\sqrt{\eps}}=2^{k-2}$ different slopes.
    
    Since $W$ is well-behaved, there are at least $\frac{1}{16\sqrt{\eps}}$ such paths crossing the slit of $W$ and each slope interval has at most $2^{k-2}$ non-skewed non-adventurous crossing ellipses, there must be at least $2^{-k-2}/\sqrt{\eps}$ slope intervals that contain some non-adventurous non-skewed ellipses through the slit of~$W$. Consequently, there are at least $2^{-k-4}/\sqrt{\eps}$ such slope intervals among  $\mathrm{Slo}^+_{4i+\rho}$ for some fixed $\rho\in \{0,1,2,3\}$. Let $\Psi^+_W$ contain one spanner path through $W$ that is non-adventurous and non-skewed from each slope interval $\mathrm{Slo}^+_{4i+\rho}$ if there is such a path. We define $\Psi^-_W$ analogously. Notice that we immediately get $|\Psi^+_W|\geq 2\cdot 2^{-k-4}/\sqrt{\eps}=\Omega(1/\sqrt{\eps})$.
    
    Moreover, observe that due to slope restrictions, any path in $\Psi^+_W$ enters $W$ on either the bottom half of the left side of $W$ or the left half of the bottom side of $W$, and exits on either the top half of the right side or the right half of the top side of $W$ (see \Cref{fig:windowslit}(ii)). The symmetric claim holds for any path in $\Psi^-_W$, thus any pair of such paths must have their intersection inside $W$, concluding the proof of (i).

    Next, we prove that for every pair of distinct paths $\gamma,\gamma'\in \Psi^+_W \cup \Psi^-_W$, the corresponding tubes $R_W(\gamma)$ and $R_W(\gamma')$ have sufficiently different angles and that their intersection has small diameter.
    Observe that the slopes in $\mathrm{Slo}^+_{4i+\rho}$ and  $\mathrm{Slo}^+_{4i+4+\rho}$ differ by at least three times the size of $\mathrm{Slo}_i$, thus the slopes of any two ellipses with slope from $\mathrm{Slo}^+_{4i+\rho}$ and $\mathrm{Slo}^+_{4i+4+\rho}$ must differ by at least $3\cdot 2^k/\sqrt{\eps}$. In particular, this holds for any pair of distinct spanner paths $\gamma,\gamma'\in \Psi^+_W \cup \Psi^-_W$. Because $\gamma$ is non-adventurous, it is contained in a tube $R_W(\gamma)$ of width $2^k\eps$. Furthermore, because $\gamma$ is non-skewed, its slope differs from the slope of $R_W(\gamma)$ by at most $2^k\sqrt{\eps}$. Similarly, $\gamma'$ is contained in a tube $R_W(\gamma')$ of width $2^k\sqrt{\eps}$ whose slope differs from the slope of $\gamma'$ by at most $2^k\sqrt{\eps}$. Hence, the slopes of $R_W(\gamma)$ and $R_W(\gamma')$ differ by at least $3\cdot 2^k\sqrt{\eps} - 2\cdot 2^k\sqrt{\eps}=2^k\sqrt{\eps}$, concluding the proof of~(ii).

    Notice first that in case $R_W(\gamma)$ and $R_W(\gamma')$ are disjoint the claim (iii) trivially holds, so assume that they intersect.
    Consequently, $R_W(\gamma)\cap R_W(\gamma')$ is a rhombus where the distance between any two opposite sides is $w:=2^k\eps$ and the smaller angle $\alpha$ is at least $2^k\sqrt{\eps}$. Any rhombus of height $w$ and angle $\alpha<\pi/2$ of side length $s$ satisfies $w=s\cdot \sin \alpha$, and so its area can be bounded as
    \[ws=w^2/\sin(\alpha) < 2^{2k}\eps^2/\sin(2^k\sqrt{\eps}) < 2^{k+1}\eps^{3/2}\]
    when $\eps<2^{-2k}$.
    This implies the desired inequality about the area of the intersection: since $\gamma$ is a core path, the corresponding region $R_W(\gamma)$ has a central line of length least $2^9\sqrt{\eps}$, thus $R_W(\gamma)$ has area at least $2^9\sqrt{\eps}\cdot 2^k\eps>2\cdot 2^{k+1}\eps^{3/2}$.
\end{proof}

We apply \Cref{thm:FGR22} in each well-behaved window $W\in \mathcal{W}$.
  \begin{lemma}\label{cl:tubes} 
            Every well-behaved window $W\in \mathcal{W}$ contains $\Omega_\mu (1/\eps^{1-\mu})$ Steiner points. 
    \end{lemma}
    \begin{proof}
    In order to use \Cref{thm:FGR22}, we scale $W$ up to the square $[\frac38,\frac58]^2$ of side length $\frac14$, and consider the square $[0,1]^2$ as the \EMPH{frame} of the window. 
    Note the scaling factor is $\Theta(\sqrt{1/\eps})$. By \Cref{lem:essenitiallydistinct} in both the positive and the negative bundles, at least $\Omega(\sqrt{1/\eps})$ ellipses $\cE_{ab}$ cross the slit of $W$ such that the corresponding spanner paths $\gamma_{ab}$ are neither adventurous nor skewed. Each such path $\gamma_{ab}$ has a corresponding strip $R_W(\gamma_{ab})$ of width $2^k\eps=\Theta(\eps)$, which is scaled up and extended to a tube of width $\delta$, where $\delta=\Theta(\sqrt{\eps})$.
    We observe that each strip $R_W(\gamma)$ can be extended to a $\delta\times 1$ rectangle within the frame that covers $R_W(\gamma)$, thus each strip corresponds to a $\delta$-tube in $[0,1]^2$.

    Let $\mathcal{T}^+$ and $\mathcal{T}^-$ be the set of these $\delta$-tubes corresponding to the positive and negative bundles through the slit of $W$, resp., and let $\mathcal{T}=\mathcal{T}^-\cup\mathcal{T}^+$. The number of $\delta$-tubes is $|\mathcal{T}|=\Theta(\sqrt{1/\eps})$. Scaling does not change the angles: the angle between any two tubes remains $\Omega(\sqrt{1/\eps})$.
    Note that the $\delta$-tubes in $\mathcal{T}$ are essentially distinct by \Cref{lem:essenitiallydistinct}.

    We use \Cref{thm:FGR22} with parameter $\delta=\Theta(\sqrt{\eps})$, $X=\Theta(\sqrt{1/\eps})$,  $Y=1$, and $|\mathcal{T}|=\Theta(\sqrt{1/\eps})$. Observe that $|\mathcal{T}_{\max}|\cdot \Theta(\sqrt{1/\eps})$. In this setting, \Cref{thm:FGR22} states that for all $r>1/\eps^\mu$, the number of essentially distinct $r$-rich $\delta$-disks is 
    \begin{equation}\label{eq:FGR22}
    |\mathcal{B}_r (\mathcal{T})|  
    \leq O_\mu \left( |\mathcal{T}|^2 \cdot \frac{1}{r^2} \left(\frac{1}{r}+1\right)\right)
    \leq O_\mu \left(\frac{|\mathcal{T}|^2}{r^3} \right)
    \leq O_\mu \left(\frac{1}{\eps r^3} \right) .
    \end{equation}

    For an integer $r\geq \Omega_\mu(1/\eps^\mu)$, consider a set $\mathcal{B}_r(\mathcal{T})$ of essentially distinct $r$-rich $\delta$-disks.
    For every $i\in \mathbb{N}$, let $\mathcal{D}_{r,i}\subseteq \mathcal{B}_r(\mathcal{T})$ be the set of disks $D\in \mathcal{B}_r(\mathcal{T})$ that intersect at least $2^{i-1}r$ but less than $2^i r$ tubes in $\mathcal{T}$. Note that $\mathcal{B}_r(\mathcal{T})=\bigcup_{i=1}^\infty \mathcal{D}_{r,i}$;
    and $\mathcal{D}_{r,i}$ is a set of $(2^{i-1}r)$-rich disks. 
    By \Cref{eq:FGR22}, the number of pairs of tubes $(T_1,T_2)\in \mathcal{T}^2$ 
    that both intersect some disk in $\mathcal{B}_r(\mathcal{T})$ is bounded by 
    \begin{equation}\label{eq:rich}
    \sum_{i=1}^\infty  (2^i r)^2 \cdot |\mathcal{D}_{r,i} |
    \leq \sum_{i=1}^\infty  (2^i r)^2 \cdot O_\mu \left(\frac{1}{\eps\, (2^{i-1}r)^3}\right) 
    \leq \sum_{i=1}^\infty  O_\mu \left(\frac{1}{\eps\, 2^ir} 
    \right) 
    =O_\mu \left(\frac{1}{\eps\, r} 
    \right) .
    \end{equation}

    Let $S$ be the set of Steiner vertices of $G$ in the frame $[0,1]^2$. We say that a vertex $v\in S$ is $r$-rich if $v$ is a vertex of at least $r$ paths $\gamma_{ab}$ that each correspond to some tube in $\mathcal{T}$. For every $r\in \mathbb{N}$, we greedily cover the $r$-rich Steiner vertices with disks of radius $\delta$ as follows: Initialize $\mathcal{B}_r(\mathcal{T}):=\emptyset$; while there is a $r$-rich Steiner vertex $v$ that is not covered by the disks in $\mathcal{B}_r(\mathcal{T})$, add a new $\delta$-disk centered at $v$ to $\mathcal{B}_r(\mathcal{T})$. By construction, the distance between the centers of any two disks in $\mathcal{B}_r(\mathcal{T})$ is at least $\delta$, and so the disks in $\mathcal{B}_r(\mathcal{T})$ are essentially distinct. 

    By construction, every tube in $\mathcal{T}^-$ crosses every tube in $\mathcal{T}^+$, and so there are at least $|\mathcal{T}^-|\cdot |\mathcal{T}^+|=\Theta(1/\eps)$ pairs of tubes in $\mathcal{T}^-\times \mathcal{T}^+$ that cross in $[0,1]^2$.
    By \Cref{eq:rich}, there exists an integer $r_0= \Theta_\mu(1/\eps^\mu)$ such that at most $\frac12\cdot |\mathcal{T}^-|\cdot |\mathcal{T}^+|$ pairs of tubes in $\mathcal{T}\times \mathcal{T}$ 
    are incident to a common $r$-rich disk in $\mathcal{B}_r(\mathcal{T})$.

    This means that there are at least $\frac12\cdot |\mathcal{T}^-|\cdot |\mathcal{T}^+|$ 
    pairs of tubes $(T_1,T_2)\in \mathcal{T}^-\times \mathcal{T}^+$ such that 
    $T_1\cap T_2$ is disjoint from all $\delta$-disks in $\mathcal{B}_r(\mathcal{T})$.
    Consequently, $T_1\cap T_2$ contains a Steiner vertex that is incident to less than $r_0$ paths $\gamma_{ab}$ corresponding to tubes in $\mathcal{T}$.

   Let $S_0\subset S$ be the set of Steiner vertices that are each incident to less than $r_0$ paths $\gamma_{ab}$ corresponding to tubes in $\mathcal{T}$.
   Counting the pairs of tubes that pass through a common Steiner vertex in $S_0$, we have 
   \begin{align*}
   \sum_{s\in S_0} \binom{r_0}{2} & \geq \frac12\cdot |\mathcal{T}^-|\cdot |\mathcal{T}^+|\\
   |S_0| \cdot r_0^2 &\geq \Theta \left(\frac{1}{\eps}\right).
    \end{align*}    
    Combined with $r_0= O_\mu(1/\eps^\mu)$, we obtain 
    $|S|\geq |S_0|\geq \Omega_\mu(1/\eps^{1-2\mu})$. Scaling the parameter $\mu$ by a factor of 2 yields $|S|\geq \Omega_\mu(1/\eps^{1-\mu})$, as required.  
    \end{proof}
\lowerbound*
\begin{proof}
    Consider $n\sqrt{\eps}/2$ copies of an axis-aligned unit square with the point sets $A,B$ on their left and right hand sides from the basic construction, and suppose that these copies have minimum distance at least $2$. Observe that for $\eps<0.1$ the spanner paths in each copy must remain disjoint. The resulting point set has exactly $\frac{n\sqrt{\eps}}{2}\cdot {2/\sqrt{\eps}}=n$ points. In each copy we have at least $|\cW|/2=\Omega(1/\eps)$ well-behaved windows by \Cref{lem:wellbehaved}. Thus, \Cref{cl:tubes} implies that for any $\mu>0$, there are at least $\Omega(1/\eps)\cdot \Omega_\mu(1/\eps^{1-\mu})=\Omega_\mu(1/\eps^{2-\mu})$ Steiner points in each copy of the unit square. Overall, any noncrossing Steiner $(1+\eps)$-spanner must have $\frac{n\sqrt{\eps}}{2}\cdot \Omega_\mu(1/\eps^{2-\mu})=\Omega_\mu(n/\eps^{3/2-\mu})$ Steiner vertices, as claimed. 
\end{proof}

\section{Conditional lower bound for cone-restricted noncrossing spanners}
\label{sec:cone}

In this section, we prove the lower bound in \Cref{thm:cone}. The core construction is a point set $P=A\cup B$, where $A$ and $B$ consist of $\Theta(\eps^{-1/2})$ equally spaced points on two opposite sides of a unit square $[0,1]^2$. Then a cone-restricted plane Steiner $(1+\eps)$-spanner $G$ contains $ab$-paths $\gamma_{ab}$ of length at most $(1+\eps)\, |ab|$, where all edges $e$ make an angle at most $\sqrt{\eps}$ with $ab$. By Lemma~\ref{lem:angles}, for all $a,a'\in A$ and $b,b'\in B$, if $ab\neq a'b'$, then the paths $\gamma_{ab}$ and $\gamma_{a'b'}$ are edge-disjoint, and if they cross (at a Steiner point), then they cross transversally. Then we consider the Steiner points in $G$ where two or more paths $\gamma_{ab}$ cross each other, and use a lower bound on degenerate crossing numbers~\cite{AckermanP13,PachT09} to derive a lower bound on the number of Steiner points. 

\paragraph{Degenerate Crossing Number.}
Given a graph $G=(V,E)$, the (classical) \textbf{crossing number} ${\rm cr}(G)$ of $G$ is the minimum number of edge crossings in a \emph{good drawing} of $G$: A good drawing of $G$ maps the vertices $v\in V$ to distinct points in the plane, and the edges $e\in E$ to Jordan arcs between the points corresponding to the endpoints of $e$ such that (i) the relative interior of any two Jordan arcs have finitely many intersection points, each of which is a transversal crossing, (ii) no vertex lies in the relative interior of a Jordan arc, and (iii) no point lies in the relative interior of three or more Jordan arcs. A \textbf{crossing} is an intersection point between two Jordan arcs that is not a vertex. The celebrated \emph{Crossing Lemma} states ${\rm cr}(G)\geq \Omega(|E|^3/|V|^2)$ if $|E|\geq 4|V|$.
The \textbf{degenerate crossing number}, denoted ${\rm dcr}(G)$, of a graph $G$ is defined analogously to the (classical) crossing number, except that condition (iii) of a good drawing is dropped~\cite{PachT09}: Arbitrarily many edges can cross (transversely) at any given crossing. If $k$ edges pairwise cross at a point $p\in \mathbb{R}^2$, then we could perturb the Jordan arcs to obtain a good drawing with $\binom{k}{2}$ pairwise crossings---but this $k$-fold crossing (i.e., a degenerate crossing) is counted only once in the definition of ${\rm dcr}(G)$. Importantly,  Ackerman et al.~\cite{AckermanP13} proved that
${\rm dcr}(G)\geq \Omega(|E|^3/|V|^2)$ if $|E|\geq 4|V|$, that is, the classical Crossing Lemma generalizes to degenerate crossings.

\begin{lemma}\label{lem:LB:cone}
     For every sufficiently small $\eps>0$ and every $n\in \mathbb{N}$, there exists a set of $n$ points in Euclidean plane for which every cone-restricted noncrossing Steiner $(1+\eps)$-spanner has $\Omega(n/\varepsilon^{3/2})$ Steiner vertices. 
\end{lemma}
\begin{proof}
Consider the basic construction $P=A\cup B$ defined above. 
Let $G$ be a cone-restricted plane Steiner $(1+\eps)$-spanner for $P$. 
For every pair $(a,b)\in A\times B$, let $\gamma_{ab}$ be a shortest path in $G$ in which every edge $e$ satisfies $\angle (e,ab)\leq \sqrt{\eps}$. We may also assume that shortest paths in $G$ are unique (by artificially perturbing the edge weights). Since $G$ is a $(1+\eps)$-spanner, then $|\gamma_{ab}|\leq (1+\eps)\, |ab|$.
In particular, as noted in \Cref{sec:pre}, $\gamma_{ab}$ lies in the ellipse $\mathcal{E}_{ab}$, with foci $a,b$ and major axis $(1+\eps)\, |ab|$. 

\Cref{lem:ellipses,lem:angles} imply that for any $a,a'\in A$ and $b,b'\in B$, the paths $\gamma_{ab}$ and $\gamma_{a'b'}$ are disjoint if $ab$ and $a'b'$ are parallel, and $\gamma_{ab}$ and $\gamma_{a'b'}$ are edge-disjoint otherwise. 
We claim, furthermore, that if $(a,b)\neq (a'b')$, then $\gamma_{ab}$ and $\gamma_{a'b'}$ share at most one vertex, and any common vertex is either a common endpoint ($a=a'$ or $b=b'$) or a point where $\gamma_{ab}$ and $\gamma_{a'b'}$ cross transversely.   
(1) First, assume that $\gamma_{ab}$ and $\gamma_{a'b'}$ have a common endpoint, w.l.o.g., $a=a'$. Then $\gamma_{ab}\subset \lozenge_{ab}$ and $\gamma_{ab'}\subset \lozenge_{ab'}$. Since $\lozenge_{ab}\cap \lozenge_{ab'}=\{a\}$, then $\gamma_{ab}$ and $\gamma_{a'b'}$ share only the vertex $a=a'$. 
(2) Next assume that $\gamma_{ab}$ and $\gamma_{a'b'}$ share a Steiner vertex $v$.
Since $a\neq a'$, $b\neq b'$, and $\mathcal{E}_{ab}\cap \mathcal{E}_{a'b'}\neq \emptyset$, then \Cref{lem:ellipses} implies that $ab$ and $a'b'$ are nonparallel.
By construction, this means that $\angle (\overrightarrow{ab},\overrightarrow{a'b'})>2\pi$. Note that $v$ decomposes $\gamma_{ab}$ into subpaths $\gamma_{ab}[a,v]$ and $\gamma_{ab}[v,b]$; and similarly $\gamma_{a'b'}=\gamma_{a'b'}[a',v]\cup \gamma_{a'b'}[v,b']$.
These four subpaths lie in cones with apex $v$ and aperture $2\sqrt{\eps}$ centered at rays of directions $\overrightarrow{ab}$, $\overrightarrow{ba}$, $\overrightarrow{a'b'}$, and $\overrightarrow{b'a'}$, respectively. 
Since any two cones intersect  only in $v$, then $\gamma_{ab}$ and $\gamma_{a'b'}$ do not share any other vertex, and they cross transversely at $v$.

Now $G$ yields a drawing of the complete bipartite graph $K(A,B)$ on partite sets $A$ and $B$ as follows: For every $a\in A$ and $b\in B$, the edge $ab$ is drawn as the path $\gamma_{ab}$. This drawing satisfies properties (i) and (ii) of good drawings; but it need not satisfy property (iii) as three or more paths $\gamma_{ab}$ may pairwise cross at a Steiner vertex of $G$. Let $s$ be the number of Steiner vertices of $G$ where two or more paths $\gamma_{ab}$ cross each other. Note that $K(A,B)$ has $|A|\cdot |B|=\Theta(\eps^{-1})$ edges and $|V|=|A|+|B|=\Theta(\eps^{-1/2})$ vertices. By the lower bound on the degenerate crossing number, we have 
$s\geq \Omega((|A|\cdot |B|)^3/(|A|+|B|)^2) = \Omega(\eps^{-2}) = \Omega(\eps^{3/2})\cdot |V|$. 

In the general case, we are given $\eps>0$ and $n\in \mathbb{N}$. We assume w.l.o.g.\ that $\eps=4^{-k}$ for some $k\in \mathbb{N}$. Our basic construction above used $|P_0|=2(2^{k-2}+1)=\Theta(\eps^{1/2})$ points. In general, we construct $P$ as $\lfloor n/|P_0|\rfloor = \Theta(\sqrt{\eps}\, n)$ disjoint copies of $P_0$, where the unit squares are aligned along the $x$-axis at unit distance apart. 
Importantly, the ellipses $\mathcal{E}_{ab}$ corresponding to point pairs $ab$ in disjoint copies of $P_0$ are disjoint. Since the Steiner points are in the paths $\gamma_{ab}\subset \mathcal{E}_{ab}$,
then the set of Steiner vertices needed for distinct copies of $P_0$ are disjoint.  
Since each copy of $P_0$ requires $\Omega(\eps^{-2})$ Steiner vertices, 
then $\Theta(\sqrt{\eps}\, n)$ copies require $\Theta(\eps^{-2}\cdot \sqrt{\eps}\, n)=\Theta(\eps^{3/2}\, n)$ Steiner vertices. 
\end{proof}

\section{Conclusion and Future Directions}

We studied Problem~\ref{problem:1} on the minimum number of Steiner points required to construct a noncrossing Euclidean $(1+\varepsilon)$-spanner, thereby almost resolving Open Problem~17 from the survey of Bose and Smid~\cite{BoseS13}. Our lower and upper bounds, $\Omega_\mu(n/\varepsilon^{3/2-\mu})\leq s(n,\varepsilon)\leq O(n/\varepsilon^{3/2})$, match up to a subpolynomial factor in $\eps$.

An important direction for future work is to understand the trade-off between planarity, \emph{total weight}, and the number of Steiner points. In particular, known constructions have not been analyzed in terms of lightness, and it remains open whether one can achieve near-optimal weight while retaining asymptotically optimal sparsity in noncrossing Steiner $(1+\varepsilon)$-spanners.

As discussed in \Cref{sec:intro}, the number of \emph{edge crossings} is another structural measure of complexity for geometric (\emph{non-Steiner}) spanners. It remains an open problem whether a tighter analysis of the number of crossings in a (greedy) $(1+\eps)$-spanner could potentially match our upper bound $O(n/\varepsilon^{3/2})$: 
   Determine the minimum $c(n,\eps)$ such that every set of $n$ points in the plane admits a (non-Steiner) $(1+\eps)$-spanner with at most $c(n,\eps)$ edge crossings.

We observe that the $(1+\varepsilon)$-stretch constraint is essential for obtaining meaningful bounds. 
If the stretch requirement is relaxed and 
one merely requires that each $ab$-path lies inside the ellipse $\mathcal{E}_{ab}$, 
then for the basic example (\Cref{sec:pre}) with $n=\Theta(\sqrt{1/\eps})$ points, a grid of side length $\sqrt{\varepsilon}$ already yields a solution with $n/\sqrt{\varepsilon}$ Steiner points, but the stretch increases to $\sqrt{2}$. The approximation of straight-line segments using grid paths that satisfy additional nondegeneracy conditions is related to the literature on \emph{digital segments}~\cite{ChiuKST22,ChristPS12}.

Finally, the connection to disk-tube incidences raises another intriguing open problem: How many Steiner points are required to construct a noncrossing geometric graph for a set $P$ of $n$ points in the plane if, instead of the spanner condition, for every $a,b\in P$, one requires an $ab$-path whose Fr\'echet distance from the straight-line segment $ab$ is at most $\varepsilon$?

\bibliographystyle{plainurl}
\bibliography{refs}

\end{document}